\documentclass[a4paper]{llncs}
\usepackage{longtable}
\usepackage{proof}
\usepackage[table]{xcolor}
\usepackage{mathtools}
\usepackage{setspace}
\usepackage{cite}
\usepackage{color}
\usepackage{float}
\usepackage{enumitem}
\usepackage{calligra}
\usepackage{algorithm,algpseudocode}

\makeatletter
\newcommand{\algmargin}{\the\ALG@thistlm}
\makeatother
\newlength{\whilewidth}
\algdef{SE}[parWHILE]{parWhile}{EndparWhile}[1]
{\parbox[t]{\dimexpr\linewidth-\algmargin}{%
		\hangindent\whilewidth\strut\algorithmicwhile\ #1\ \algorithmicdo\strut}}{\algorithmicend\ \algorithmicwhile}%
\algnewcommand{\parState}[1]{\State%
	\parbox[t]{\dimexpr\linewidth-\algmargin}{\strut #1\strut}}

\DeclareFontShape{T1}{calligra}{m}{n}{<->s*[2.2]callig15}{}
\DeclareMathAlphabet{\matcalligra}{T1}{calligra}{m}{n}

\DeclareMathAlphabet\mathbfcal{OMS}{cmsy}{b}{n}
\let\mathcal\undefined \DeclareMathAlphabet{\mathcal}{OMS}{cmsy}{m}{n}
\newcommand{\flqs}{\ensuremath{\mathsf{4LQS}}}
\newcommand{\flqsr}{\ensuremath{\mathsf{4LQS^R}}}
\newcommand{\dlss}{\mathcal{DL}\langle \mathsf{4LQS^R}\rangle(\D)}
\newcommand{\dlssx}{\mathcal{DL}\langle \mathsf{4LQS^{R,\!\times}}\rangle(\D)}
\newcommand{\shdlssx}{\mathcal{DL}_{\D}^{4,\!\times}}
\newcommand{\shdlss}{\mathcal{DL}_{\D}^{4}}
\newcommand{\D}{\mathbf{D}}
\newcommand{\sroiqd}{\mathcal{SROIQ}(\D)}

\newcommand{\defAs}{\coloneqq}
\newcommand{\I}{\mathbf{I}}
\newcommand{\Ind}{\mathbf{Ind}}
\newcommand{\C}{\mathbf{C}}

\newcommand{\Ra}{\mathbf{R_A}}
\newcommand{\Rd}{\mathbf{R_D}}
\newcommand{\sym}{\mathsf{Sym}}
\newcommand{\asym}{\mathsf{Asym}}
\newcommand{\refl}{\mathsf{Ref}}
\newcommand{\irref}{\mathsf{Irref}}
\newcommand{\tra}{\mathsf{Tra}}
\newcommand{\fun}{\mathsf{Fun}}
\newcommand{\ck}{\mathsf{cpt}_\mathcal{KB}}
\newcommand{\ark}{\mathsf{arl}_\mathcal{KB}}
\newcommand{\crk}{\mathsf{crl}_\mathcal{KB}}
\newcommand{\ik}{\mathsf{ind}_\mathcal{KB}}

\newcommand{\bfk}{\mathsf{bf}_\mathcal{KB}^{\D}}
\newcommand{\vipcomment}[1]{}
\newcommand{\pow}{\mathcal{P}}

\newcommand{\T}{\mathcal{T}}

\newcommand{\seq}{\mathcal{S}^{\overline{\beta}}_i}

\newcommand{\seqnj}{\mathcal{S}^{\overline{\beta}}_j}
\newcommand{\ke}{KE-tableau}
\newcommand{\M}{\mathbfcal{M}}

\newcommand{\KB}{\mathcal{KB}}
\newcommand{\Vars}{\mathtt{Vars}}
\newcommand{\vari}{\mathtt{Var}_i}
\newcommand{\varz}{\mathtt{Var}_0}
\newcommand{\varu}{\mathtt{Var}_1}

\newcommand{\vart}{\mathtt{Var}_3}

\newcommand{\sfvar}[2]{ {\mathsf{#1}_{#2}} }
\newcommand{\varar}{\mathsf{V}_{\mathsf{ar}}}
\newcommand{\varcon}{\mathsf{V}_{\mathsf{c}}}
\newcommand{\varind}{\mathsf{V}_{\mathsf{i}}}
\newcommand{\varcr}{\mathsf{V}_{\mathsf{cr}}}

\newcommand{\DT}{\mathcal{D}}

\makeatletter
\def\@seccntformat#1{\@ifundefined{#1@cntformat}%
   {\csname the#1\endcsname\quad}  
   {\csname #1@cntformat\endcsname}
}
\let\oldappendix\appendix 
\renewcommand\appendix{%
    \oldappendix
    \newcommand{\section@cntformat}{\appendixname~\thesection\quad}
}

\makeatother

\makeatletter
\renewcommand\subsubsection{\@startsection{subsubsection}{3}{\z@}%
                       {-18\p@ \@plus -4\p@ \@minus -4\p@}%
                       {0.5em \@plus 0.22em \@minus 0.1em}%
                       {\normalfont\normalsize\bfseries\boldmath}}
\makeatother
\title{A set-theoretic approach to ABox reasoning services (Extended Version)} 

\author{Domenico Cantone \and Marianna Nicolosi-Asmundo \and \\Daniele Francesco Santamaria}
\institute{
	University of Catania, Dept. of Mathematics and Computer Science\\
	~email:~\texttt{\{cantone,nicolosi,santamaria\}@dmi.unict.it}
}

\begin{document}
	\maketitle
	
	
	\begin{abstract}
	In this paper we consider the most common ABox reasoning services for the description logic  $\dlssx$ ($\shdlssx$, for short) and prove their decidability via a reduction to the satisfiability problem for the set-theoretic fragment \flqsr. The description logic $\shdlssx$ is very expressive, as it admits various concept and role constructs, and data types, 
	that allow one to represent rule-based languages such as SWRL. 
	
	Decidability results are achieved by defining a generalization of the conjunctive query answering problem, called HOCQA (Higher Order Conjunctive Query Answering), that can be instantiated to the most wide\-spread ABox reasoning tasks. We also present a \ke\space based procedure for calculating the answer set from $\shdlssx$ knowledge bases and higher order $\shdlssx$ conjunctive queries, thus providing means for reasoning on several well-known ABox reasoning tasks.
	Our calculus extends a previously introduced \ke\space based decision procedure for the CQA problem.
	\end{abstract}


\section{Introduction}
Recently, results from Computable Set Theory have been applied to knowledge representation for the semantic web in order to define and reason about description logics and rule languages. Such a study is motivated by the fact that Computable Set Theory is a research field plenty of interesting decidability results and that there exists a natural translation function between some set theoretical fragments and description logics and rule languages.
 
%
In particular, the decidable four-level stratified fragment of set theory $\flqsr$, involving variables of four sorts, pair terms, and a restricted form of quantification over variables of the first three sorts (cf.\ \cite{CanNic2013}), has been used in \cite{CanLonNicSanRR2015} to represent the description logic $\dlss$ (more simply referred to as $\shdlss$). The logic $\shdlss$ admits concept constructs such as full negation, union and intersection of concepts, concept domain and range, existential quantification and min cardinality on the left-hand side of inclusion axioms. It also supports role constructs such as role chains on the left hand side of inclusion axioms, union, intersection, and complement of abstract roles, and properties on roles such as transitivity, symmetry, reflexivity, and irreflexivity. As briefly shown in \cite{CanLonNicSanRR2015}, $\shdlss$ is particularly suitable to express a rule language such as the Semantic Web Rule Language (SWRL), an extension of the Ontology Web Language (OWL). It admits
data types, a simple form of concrete domains that are relevant in real world applications.  
In \cite{CanLonNicSanRR2015}, the consistency problem for $\shdlss$-knowledge bases has been proved decidable by means of a reduction to the satisfiability problem for $\flqsr$, whose decidability has been established in \cite{CanNic2013}. It has also been shown that, under not very restrictive constraints, the consistency problem for $\shdlss$-knowledge bases is \textbf{NP}-complete. Such a low complexity result is motivated by the fact that existential quantification cannot appear on the right-hand side of inclusion axioms. Nonetheless, $\shdlss$ turns out to be  more expressive than other low complexity logics such as OWL RL and suitable for representing real world ontologies. 
For example the restricted version of $\shdlss$ allows one to express  several ontologies, such as \textsf{Ontoceramic}  \cite{cilc15} classifying ancient pottery. 


In \cite{ictcs16}, the description logic $\dlssx$ ($\shdlssx$, for short), extending $\shdlss$ with Boolean operations on concrete roles and with the product of concepts, has been introduced and the \emph{Conjunctive Query Answering} (CQA) problem for $\shdlssx$ has been proved decidable via a reduction to the CQA problem for $\flqsr$, whose decidability follows from that of $\flqsr$ (see \cite{CanNic2013}). CQA is a powerful way to query ABoxes, particularly relevant in the context of description logics and for real world applications based on semantic web technologies, as it provides mechanisms for interacting with ontologies and data. The CQA problem for description logics has been introduced in \cite{calvanese98,calvanese08} and studied for several well-known description logics (cf.\ \cite{Calvanese:1998:DQC:275487.275504, Calvanese2013335, calvanese2007answering, Ortiz:Calvanese:et-al:06a, GlHS07a,
GliHoLuSa-JAIR08,HorrSattTob-CADE-2000,j.websem63, HorrTess-aaai-2000,DBLP:conf/ijcai/HustadtMS05, Rosa07c, DBLP:conf/cade/Lutz08, Ortiz:2011:QAH:2283516.2283571}). Finally, we mention also a terminating \ke\space based procedure that, given a $\shdlssx$-query $Q$ and a $\shdlssx$-knowledge base $\mathcal{KB}$ represented in set-theoretic terms, determines the answer set of $Q$ with respect to $\mathcal{KB}$. \ke\space systems \cite{dagostino1999} allow the construction of trees whose distinct branches define mutually exclusive situations, thus preventing the proliferation of redundant branches, typical of semantic tableaux.  

In this paper we extend the results presented in \cite{ictcs16} by considering also the main ABox reasoning tasks for $\shdlssx$, such as instance checking and concept retrieval, and study their decidability via a reduction to the satisfiability problem for  $\flqsr$. Specifically, we define Higher Order (HO) $\shdlssx$-conjunctive queries  admitting variables of three sorts: individual and data type values variables, concept variables, and role variables. HO $\shdlssx$-conjunctive queries can be instantiated to any of the ABox reasoning tasks we are considering in the paper. Then, we define the Higher Order Conjunctive Query Answering (HOCQA) problem for $\shdlssx$ and prove its decidability by reducing it to the HOCQA problem for $\flqsr$. Decidability of the latter problem follows from that of the satisfiability problem for $\flqsr$. $\flqsr$ representation of $\shdlssx$ knowledge bases is defined according to \cite{ictcs16}. $\flqsr$ turns out to be naturally suited for the HOCQA problem since
HO $\shdlssx$-conjunctive queries are easily translated into $\flqsr$-formulae. In particular, individual and data type value variables are mapped into $\flqsr$ variables of sort 0, concept variables into $\flqsr$ variables of sort 1, and role variables into $\flqsr$ variables of sort 3. Finally, we present an extension of the \ke\space presented in \cite{ictcs16}, which provides a decision procedure for the HOCQA task for $\shdlssx$.


\section{Preliminaries}
\subsection{The set-theoretic fragment \flqsr} \label{4LQS}

It is convenient to first introduce the syntax and semantics of a more general four-level quantified language, denoted $\flqs$.
Then 
we provide some restrictions on the quantified formulae of $\flqs$ to characterize \flqsr. The interested reader can find more details in \cite{CanNic2013} together with the decision procedure for the satisfiability problem for \flqsr.
%
%

$\flqs$ involves four collections, $\mathcal{V}_i$, of variables of sort $i=0,1,2,3$, respectively. These will be denoted by $X^i,Y^i,Z^i,\ldots$ (in particular, variables of sort $0$ will also be denoted by $x, y, z, \ldots$).
In addition to variables, $\flqs$ involves also \emph{pair terms} of the form $\langle x,y \rangle$, for $ x,y \in \mathcal{V}_0$.

\smallskip

\noindent\emph{$\flqs$-quantifier-free atomic formulae} are classified as:
\begin{itemize}[topsep=0.1cm, itemsep=0.1cm]
\item[-] level 0:~~ $x=y$,~~ $x \in X^1$,~~ $\langle x,y \rangle = X^2$,~~ $\langle x,y \rangle \in X^3$;
\item[-] level 1:~~ $X^1=Y^1$,~~ $X^1 \in X^2$;
\item[-] level 2:~~ $X^2=Y^2$,~~ $X^2 \in X^3$.
\end{itemize}

\smallskip

\noindent $\flqs$-\emph{purely universal formulae} are classified as:
\begin{itemize}[topsep=0.1cm, itemsep=0.1cm]
\item[-] { level 1: $(\forall z_1)\ldots(\forall z_n) \varphi _0$, where $z_1,\ldots,z_n$  $\in \mathcal{V}_0$ and $\varphi _0$ is any propositional combination of quantifier-free atomic formulae of level 0;}
\item[-] { level 2: $(\forall Z^1_1)\ldots(\forall Z^1_m) \varphi _1$, where $Z^1_1,\ldots,Z^1_m $  $\in \mathcal{V}_1$ and $\varphi _1$ is any propositional combination of quantifier-free atomic formulae of levels 0 and 1, and of purely universal formulae of level 1;}
\item[-] {level 3: $(\forall Z^2_1)\ldots(\forall Z^2_p) \varphi _2$, where $Z^2_1,\ldots,Z^2_p $  $\in \mathcal{V}_2$ and $\varphi _2$ is any propositional combination of quantifier-free atomic formulae and of purely universal formulae of levels 1 and 2.}
\end{itemize}

\noindent
$\flqs$-formulae are all the propositional combinations of quantifier-free atomic formulae of levels 0, 1, 2, and of purely universal formulae of levels 1, 2, 3.

\medskip


The variables $z_1,\ldots,z_n$ are said to occur \textit{quantified} in $(\forall z_1) \ldots (\forall z_n) \varphi_0$. Likewise, $Z^1_1,\ldots, Z^1_m$ and $Z^2_1, \ldots, Z^2_p$ occur quantified in $(\forall Z^1_1) \ldots (\forall Z^1_m) \varphi_1$ and in $(\forall Z^2_1) \ldots (\forall Z^2_p)  \varphi_2$, respectively.
A variable occurs \textit{free} in a $\flqs$-formula $\varphi$
if it does not occur quantified in any subformula of $\varphi$. For $i = 0,1,2,3$, we denote with $\vari(\varphi)$ the collections of variables of level $i$ occurring free in $\varphi$ and we put  $\Vars(\varphi) \defAs \bigcup_{i=0}^{3} \vari(\varphi)$.

A 
substitution $\sigma \defAs \{ \vec{x}/\vec{y}, \vec{X}^1/\vec{Y}^1, \vec{X}^2/\vec{Y}^2, \vec{X}^3/\vec{Y}^3\}$ is the mapping $\varphi \mapsto \varphi\sigma$ such that, for any given $\flqs$-formula $\varphi$, $\varphi\sigma$ is the $\flqs$-formula obtained from $\varphi$ by replacing the free occurrences of the variables $x_i$ in $\vec{x}$ (for $i = 1,\ldots, n$)  
with the corresponding $y_i$ in $\vec{y}$, of $X^1_j$ in $\vec{X}^1$ (for $j = 1,\ldots,m$) with $Y^1_j$ in $\vec{Y}^1$, of $X^2_k$ in $\vec{X}^2$ (for $k = 1,\ldots,p$) with $Y^2_k$ in $\vec{Y}^2$, and of $X^3_h$ in $\vec{X}^3$ (for $h= 1,\ldots,q$) with $Y^3_h$ in $\vec{Y}^3$, respectively.  A substitution $\sigma$ is \emph{free} for $\varphi$ if the formulae $\varphi$ and $\varphi\sigma$ have exactly the same occurrences of quantified variables. The \emph{empty substitution}, denoted with $\epsilon$, satisfies $\varphi \epsilon = \varphi$, for every $\flqs$-formula $\varphi$.



A $\flqs$-\emph{interpretation} is a pair $\mathbfcal{M}=(D,M)$, where $D$ is a non-empty collection of objects (called \emph{domain} or \emph{universe} of $\mathbfcal{M}$) and $M$ is an assignment over the variables in $\mathcal{V}_i$, for $i=0,1,2,3$,  such that:\\[.1cm]
\centerline{$MX^{0} \in D, ~~~  MX^1 \in \pow(D), ~~~ MX^2 \in \pow(\pow(D)), ~~~ MX^3 \in \pow(\pow(\pow(D))),$}\\[.1cm]
where $ X^{i} \in \mathcal{V}_i$, for $i=0,1,2,3$, and $\pow(s)$ denotes the powerset of $s$.
%
%
%
%
%

\smallskip
\noindent
Pair terms are interpreted \emph{\`a la} Kuratowski, and therefore we put \\[.1cm]
\centerline{$M \langle x,y \rangle \defAs \{ \{ Mx \},\{ Mx,My \} \}$.}\\[0.2cm]
%
%
%
Quantifier-free atomic formulae and purely universal formulae are evaluated in a standard way according to the usual meaning of the predicates `$\in$'
and `$=$'. The interpretation of quantifier-free atomic formulae and of purely universal formulae is given in \cite{CanNic2013}.

%
%

Finally, compound formulae are interpreted according to the standard rules of propositional logic. If $\mathbfcal{M} \models \varphi$, then $\mathbfcal{M} $ is said to be a $\flqs$-model for $\varphi$. A $\flqs$-formula is said to be \emph{satisfiable} if it has a $\flqs$-model. A $\flqs$-formula is \emph{valid} if it is satisfied by all $\flqs$-interpretations. 

We are now ready to present the fragment \flqsr\ of $\flqs$ of our interest. This is the collection of the formulae $\psi$ of $\flqs$ fulfilling the restrictions:
\begin{enumerate}[label=\arabic*., topsep=0.1cm, itemsep=0.1cm]
\item for every purely universal formula $(\forall Z^1_1)\ldots(\forall Z^1_m) \varphi_1$ of level 2 occurring in $\psi$ and every purely universal formula $(\forall z_1)\ldots(\forall z_n) \varphi_0$ of level 1 occurring negatively in $\varphi_1$, $\varphi_0$ is a propositional combination of quantifier-free atomic formulae of level $0$ and the condition\\[.1cm]
\centerline{$\neg \varphi_0 \rightarrow \overset{n}{ \underset {i=1} \bigwedge} \; \overset {m} { \underset {j=1 }\bigwedge} z_i \in Z^1_j$}\\[.1cm]
%
is a valid $\flqs$-formula (in this case we say that $(\forall z_1)\ldots(\forall z_n) \varphi_0$ is \emph{linked to the variables} $Z^1_1,\ldots,Z^1_m$);

\item for every purely universal formula  $(\forall Z^2_1)\ldots(\forall Z^2_p) \varphi_2$  of level 3 in $\psi$:
\begin{itemize}[topsep=0.1cm, itemsep=0.cm]
\item[-] every purely universal formula of level 1 occurring negatively in $\varphi_2$ and not occurring in a purely universal formula of level 2 is only allowed to be of the form\\[.1cm]
\centerline{$(\forall z_1)\ldots(\forall z_n) \neg( \overset {n}{ \underset {i=1} \bigwedge} \; \overset {n} { \underset {j=1}\bigwedge} \langle z_i,z_j \rangle=Y^2_{ij}),$}\\[.1cm]
with $Y^2_{ij} \in \mathcal{V}^2$, for $i,j=1,\ldots,n$;

\item[-] purely universal formulae $(\forall Z^1_1)\ldots(\forall Z^1_m) \varphi_1$ of level 2 may occur only positively in $\varphi_2$.\footnote{Definitions of positive occurrence and of negative occurrence of a formula inside another formula can be found in \cite{CanNic2013}.}
\end{itemize}

\end{enumerate}

Restriction 1 has been introduced for technical reasons concerning the decidability of the satisfiability problem for the fragment, while 
restriction  2 allows one to define binary relations and several operations on them. 

The semantics of \flqsr\space plainly coincides with that of $\flqs$.

\subsection{The logic $\dlssx$}\label{dlssx}
The description logic $\dlssx$ (which, as already remarked, will be more simply referred to as $\shdlssx$) is an extension of the  logic $\dlss$ presented in \cite{CanLonNicSanRR2015}, where Boolean operations on concrete roles and the product of concepts are defined. 
%
In addition to other features, $\shdlssx$ admits also data types, a simple form of concrete domains that are relevant in real-world applications. In particular, it treats derived data types by admitting data type terms constructed from data ranges by means of a finite number of applications of the Boolean operators. Basic and derived data types can be used inside inclusion axioms involving concrete roles.

Data types are introduced through the notion of data type map, defined  according to \cite{Motik2008} as follows. Let $\D = (N_{D}, N_{C},N_{F},\cdot^{\D})$ be a \emph{data type map}, where  $N_{D}$ is a finite set of data types, $N_{C}$ is a function assigning a set of constants $N_{C}(d)$ to each data type $d \in N_{D}$, $N_{F}$ is a function assigning a set of facets $N_{F}(d)$ to each $d \in N_{D}$, and $\cdot^{\D}$ is a function assigning a data type interpretation $d^{\D}$ to each data type $d \in N_{D}$, a facet interpretation $f^{\D} \subseteq d^{\D}$ to each facet $f \in N_{F}(d)$, and a data value $e_{d}^{\D} \in d^{\D}$ to every constant $e_{d} \in N_{C}(d)$.  We shall assume that the interpretations of the data types in $N_{D}$ are nonempty pairwise disjoint sets.

%
%
%
%

Let $\Ra$, $\Rd$, $\mathbf{C}$, $\mathbf{I}$ be denumerable pairwise disjoint sets of abstract role names, concrete role names, concept names, and individual names, respectively. We assume that the set of abstract role names $\Ra$ contains a name $U$ denoting the universal role. 
%

 \vipcomment{An abstract role hierarchy $\mathsf{R}_{a}^{H}$ is a finite collection of RIAs.  A strict partial order $\prec$ on  $\Ra \cup \{ R^- \mid R \in \Ra \}$ is called \emph{a regular order} if $\prec$ satisfies, additionally, $S \prec R$ iff $S^- \prec R$, for all roles R and S.\footnote{We recall that a strict partial order $\prec$  on a set $A$ is an irreflexive and transitive relation on $A$.}}

\noindent
(a) $\shdlssx$-data type, (b) $\shdlssx$-concept, (c) $\shdlssx$-abstract role, and (d) $\shdlssx$-concrete role terms are constructed according to the following syntax rules:
\begin{itemize}
\item[(a)] $t_1, t_2 \longrightarrow dr ~|~\neg t_1 ~|~t_1 \sqcap t_2 ~|~t_1 \sqcup t_2 ~|~\{e_{d}\}\, ,$

\item[(b)] $C_1, C_ 2 \longrightarrow A ~|~\top ~|~\bot ~|~\neg C_1 ~|~C_1 \sqcup C_2 ~|~C_1 \sqcap C_2 ~|~\{a\} ~|~\exists R.\mathit{Self}| \exists R.\{a\}| \exists P.\{e_{d}\}\, ,$

\item[(c)] $R_1, R_2 \longrightarrow S ~|~U ~|~R_1^{-} ~|~ \neg R_1 ~|~R_1 \sqcup R_2 ~|~R_1 \sqcap R_2 ~|~R_{C_1 |} ~|~R_{|C_1} ~|~R_{C_1 ~|~C_2} ~|~id(C) ~|~ $

$C_1 \times C_2    \, ,$

\item[(d)] $P_1,P_2 \longrightarrow T ~|~\neg P_1 ~|~ P_1 \sqcup P_2 ~|~ P_1 \sqcap P_2  ~|~P_{C_1 |} ~|~P_{|t_1} ~|~P_{C_1 | t_1}\, ,$
\end{itemize}
where $dr$ is a data range for $\D$, $t_1,t_2$ are data type terms, $e_{d}$ is a constant in $N_{C}(d)$, $a$ is an individual name, $A$ is a concept name, $C_1, C_2$ are $\shdlssx$-concept terms, $S$ is an abstract role name,  $R, R_1,R_2$ are $\shdlssx$-abstract role terms, $T$ is a concrete role name, and $P,P_1,P_2$ are $\shdlssx$-concrete role terms. We remark that data type terms are introduced in order to represent derived data types.

A $\shdlssx$-knowledge base is a triple ${\mathcal K} = (\mathcal{R}, \mathcal{T}, \mathcal{A})$ such that $\mathcal{R}$ is a $\shdlssx$-$RBox$, $\mathcal{T}$ is a $\shdlssx$-$TBox$, and $\mathcal{A}$ a $\shdlssx$-$ABox$. 

A $\shdlssx$-$RBox$ is a collection of statements of the following forms:\\[0.1cm]

 \centerline {$R_1 \equiv R_2$, $R_1 \sqsubseteq R_2$, $R_1\ldots R_n \sqsubseteq R_{n+1}$, $\sym(R_1)$, $\asym(R_1)$, $\refl(R_1)$,} 
 \centerline{$\irref(R_1)$, $\mathsf{Dis}(R_1,R_2)$,
$\tra(R_1)$, $\fun(R_1)$, $R_1 \equiv C_1 \times C_2$, $P_1 \equiv P_2$,} 
\centerline{$P_1 \sqsubseteq P_2$, $\mathsf{Dis}(P_1,P_2)$, $\fun(P_1)$,} 
 
  $ $\\[0.1cm] where $R_1,R_2$ are $\shdlssx$-abstract role terms, $C_1, C_2$ are $\shdlssx$-abstract concept terms, and $P_1,P_2$ are $\shdlssx$-concrete role terms. Any expression of the type $w \sqsubseteq R$, where $w$ is a finite string of $\shdlssx$-abstract role terms and $R$ is an $\shdlssx$-abstract role term, is called a \emph{role inclusion axiom (RIA)}. 

A $\shdlssx$-$TBox$ is a set of statements of the types:
\begin{itemize}
\item[-] $C_1 \equiv C_2$, $C_1 \sqsubseteq C_2$, $C_1 \sqsubseteq \forall R_1.C_2$, $\exists R_1.C_1 \sqsubseteq C_2$, $\geq_n\!\! R_1. C_1 \sqsubseteq C_2$, \\$C_1 \sqsubseteq {\leq_n\!\! R_1. C_2}$,
\item[-] $t_1 \equiv t_2$, $t_1 \sqsubseteq t_2$, $C_1 \sqsubseteq \forall P_1.t_1$, $\exists P_1.t_1 \sqsubseteq C_1$, $\geq_n\!\! P_1. t_1 \sqsubseteq C_1$, $C_1 \sqsubseteq {\leq_n\!\! P_1. t_1}$,
\end{itemize}
where $C_1,C_2$ are $\shdlssx$-concept terms, $t_1,t_2$ data type terms, $R_1$  a $\shdlssx$-abstract role term, $P_1$ a $\shdlssx$-concrete role term. Any statement of the form $C \sqsubseteq D$, with  $C$, $D$ $\shdlss$-concept terms, is a 
\emph{general concept inclusion axiom}.

A $\shdlssx$-$ABox$ is a set of \emph{individual assertions} of the forms: $a : C_1$, $(a,b) : R_1$, 
$a=b$, $a \neq b$, $e_{d} : t_1$, $(a, e_{d}) : P_1$, 
with $C_1$ a $\shdlssx$-concept term, $d$ a data type, $t_1$ a data type term, $R_1$ a $\shdlssx$-abstract role term, $P_1$ a $\shdlssx$-concrete role term, $a,b$ individual names, and $e_{d}$ a constant in $N_{C}(d)$.

The semantics of $\shdlssx$ is given by means of an interpretation $\I= (\Delta^\I, \Delta_{\D}, \cdot^\I)$, where $\Delta^\I$ and $\Delta_{\D}$ are non-empty disjoint domains such that $d^\D\subseteq \Delta_{\D}$, for every $d \in N_{D}$, and
$\cdot^\I$ is an interpretation function.
The definition of the interpretation of concepts and roles, axioms and assertions is illustrated in  Table \ref{semdlss}.
{\small
\begin{longtable}{|>{\centering}m{2.5cm}|c|>{\centering\arraybackslash}m{6.7cm}|}
\hline
Name & Syntax & Semantics \\
\hline

concept & $A$ & $ A^\I \subseteq \Delta^\I$\\

ab. (resp., cn.) rl. & $R$ (resp., $P$ )& $R^\I \subseteq \Delta^\I \times \Delta^\I$ \hspace*{0.5cm} (resp., $P^\I \subseteq \Delta^\I \times \Delta_\D$)\\


individual& $a$& $a^\I \in \Delta^\I$\\

nominal & $\{a\}$ & $\{a\}^\I = \{a^\I \}$\\

dtype  (resp., ng.) & $d$ (resp., $\neg d$)& $ d^\D \subseteq \Delta_\D$ (resp., $\Delta_\D \setminus d^\D $)\\


negative data type term & $ \neg t_1 $ & $  (\neg t_1)^{\D} = \Delta_{\D} \setminus t_1^{\D}$ \\

data type terms intersection & $ t_1 \sqcap t_2 $ & $  (t_1 \sqcap t_2)^{\D} = t_1^{\D} \cap t_2^{\D} $ \\

data type terms union & $ t_1 \sqcup t_2 $ & $  (t_1 \sqcup t_2)^{\D} = t_1^{\D} \cup t_2^{\D} $ \\

constant in $N_{C}(d)$ & $ e_{d} $ & $ e_{d}^\D \in d^\D$ \\



\hline
data range  & $\{ e_{d_1}, \ldots , e_{d_n} \}$& $\{ e_{d_1}, \ldots , e_{d_n} \}^\D = \{e_{d_1}^\D \} \cup \ldots \cup \{e_{d_n}^\D \} $ \\

data range   &  $\psi_d$ & $\psi_d^\D$\\

data range    & $\neg dr$ &  $\Delta_\D \setminus dr^\D $\\

\hline

top (resp., bot.) & $\top$ (resp., $\bot$ )& $\Delta^\I$  (resp., $\emptyset$)\\


negation & $\neg C$ & $(\neg C)^\I = \Delta^\I \setminus C$ \\

conj. (resp., disj.) & $C \sqcap D$ (resp., $C \sqcup D$)& $ (C \sqcap D)^\I = C^\I \cap D^\I$  (resp., $ (C \sqcup D)^\I = C^\I \cup D^\I$)\\


valued exist. quantification & $\exists R.{a}$ & $(\exists R.{a})^\I = \{ x \in \Delta^\I : \langle x,a^\I \rangle \in R^\I  \}$ \\

data typed exist. quantif. & $\exists P.{e_{d}}$ & $(\exists P.e_{d})^\I = \{ x \in \Delta^\I : \langle x, e^\D_{d} \rangle \in P^\I  \}$ \\

self concept & $\exists R.\mathit{Self}$ & $(\exists R.\mathit{Self})^\I = \{ x \in \Delta^\I : \langle x,x \rangle \in R^\I  \}$ \\

nominals & $\{ a_1, \ldots , a_n \}$& $\{ a_1, \ldots , a_n \}^\I = \{a_1^\I \} \cup \ldots \cup \{a_n^\I \} $ \\

\hline

universal role & U & $(U)^\I = \Delta^\I \times \Delta^\I$\\

inverse role & $R^-$ & $(R^-)^\I = \{\langle y,x \rangle  \mid \langle x,y \rangle \in R^\I\}$\\

concept cart. prod. & $ C_1 \times C_2$   &  $ (C_1 \times C_2)^I = C_1^I \times C_2^I$ \\

abstract role complement & $ \neg R $ & $ (\neg R)^\I=(\Delta^\I \times \Delta^\I) \setminus R^\I $\\

abstract role union & $R_1 \sqcup R_2$ & $ (R_1 \sqcup R_2)^\I = R_1^\I \cup R_2^\I $\\

abstract role intersection & $R_1 \sqcap R_2$ & $ (R_1 \sqcap R_2)^\I = R_1^\I \cap R_2^\I $\\

abstract role domain restr. & $R_{C \mid }$ & $ (R_{C \mid })^\I = \{ \langle x,y \rangle \in R^\I : x \in C^\I  \} $\\

concrete role complement & $ \neg P $ & $ (\neg P)^\I=(\Delta^\I \times \Delta^\D) \setminus P^\I $\\

concrete role union & $P_1 \sqcup P_2$ & $ (P_1 \sqcup P_2)^\I = P_1^\I \cup P_2^\I $\\

concrete role intersection & $P_1 \sqcap P_2$ & $ (P_1 \sqcap P_2)^\I = P_1^\I \cap P_2^\I $\\

concrete role domain restr. & $P_{C \mid }$ & $ (P_{C \mid })^\I = \{ \langle x,y \rangle \in P^\I : x \in C^\I  \} $\\

concrete role range restr. & $P_{ \mid t}$ &  $ (P_{\mid t})^\I = \{ \langle x,y \rangle \in P^\I : y \in t^\D  \} $\\

concrete role restriction & $P_{ C_1 \mid t}$ &  $ (P_{C_1 \mid t})^\I = \{ \langle x,y \rangle \in P^\I : x \in C_1^\I \wedge y \in t^\D  \} $\\

\hline

concept subsum. & $C_1 \sqsubseteq C_2$ & $\I \models_\D C_1 \sqsubseteq C_2 \; \Longleftrightarrow \; C_1^\I \subseteq C_2^\I$ \\

ab. role subsum. & $ R_1 \sqsubseteq R_2$ & $\I \models_\D R_1 \sqsubseteq R_2 \; \Longleftrightarrow \; R_1^\I \subseteq R_2^\I$\\

role incl. axiom & $R_1 \ldots R_n \sqsubseteq R$ & $\I \models_\D R_1 \ldots R_n \sqsubseteq R  \; \Longleftrightarrow \; R_1^\I\circ \ldots \circ R_n^\I \subseteq R^\I$\\
cn. role subsum. & $ P_1 \sqsubseteq P_2$ & $\I \models_\D P_1 \sqsubseteq P_2 \; \Longleftrightarrow \; P_1^\I \subseteq P_2^\I$\\

\hline

symmetric role & $\sym(R)$ & $\I \models_\D \sym(R) \; \Longleftrightarrow \; (R^-)^\I \subseteq R^\I$\\

asymmetric role & $\asym(R)$ & $\I \models_\D \asym(R) \; \Longleftrightarrow \; R^\I \cap (R^-)^\I = \emptyset $\\

transitive role & $\tra(R)$ & $\I \models_\D \tra(R) \; \Longleftrightarrow \; R^\I \circ R^\I \subseteq R^\I$\\

disj. ab. role & $\mathsf{Dis}(R_1,R_2)$ & $\I \models_\D \mathsf{Dis}(R_1,R_2) \; \Longleftrightarrow \; R_1^\I \cap R_2^\I = \emptyset$\\

reflexive role & $\refl(R)$& $\I \models_\D \refl(R) \; \Longleftrightarrow \; \{ \langle x,x \rangle \mid x \in \Delta^\I\} \subseteq R^\I$\\

irreflexive role & $\irref(R)$& $\I \models_\D \irref(R) \; \Longleftrightarrow \; R^\I \cap \{ \langle x,x \rangle \mid x \in \Delta^\I\} = \emptyset  $\\

func. ab. role & $\fun(R)$ & $\I \models_\D \fun(R) \; \Longleftrightarrow \; (R^{-})^\I \circ R^\I \subseteq  \{ \langle x,x \rangle \mid x \in \Delta^\I\}$  \\

disj. cn. role & $\mathsf{Dis}(P_1,P_2)$ & $\I \models_\D \mathsf{Dis}(P_1,P_2) \; \Longleftrightarrow \; P_1^\I \cap P_2^\I = \emptyset$\\

func. cn. role & $\fun(P)$ & $\I \models_\D \fun(p) \; \Longleftrightarrow \; \langle x,y \rangle \in P^\I \mbox{ and } \langle x,z \rangle \in P^\I \mbox{ imply } y = z$  \\

\hline

data type terms equivalence & $ t_1 \equiv t_2 $ & $ \I \models_{\D} t_1 \equiv t_2 \Longleftrightarrow t_1^{\D} = t_2^{\D}$\\

data type terms diseq. & $ t_1 \not\equiv t_2 $ & $ \I \models_{\D} t_1 \not\equiv t_2 \Longleftrightarrow t_1^{\D} \neq t_2^{\D}$\\

data type terms subsum. & $ t_1 \sqsubseteq t_2 $ &  $ \I \models_{\D} (t_1 \sqsubseteq t_2) \Longleftrightarrow t_1^{\D} \subseteq t_2^{\D} $ \\

\hline

concept assertion & $a : C_1$ & $\I \models_\D a : C_1 \; \Longleftrightarrow \; (a^\I \in C_1^\I) $ \\

agreement & $a=b$ & $\I \models_\D a=b \; \Longleftrightarrow \; a^\I=b^\I$\\

disagreement & $a \neq b$ & $\I \models_\D a \neq b  \; \Longleftrightarrow \; \neg (a^\I = b^\I)$\\


ab. role asser. & $ (a,b) : R $ & $\I \models_\D (a,b) : R \; \Longleftrightarrow \;  \langle a^\I , b^\I \rangle \in R^\I$ \\

cn. role asser. & $ (a,e_d) : P $ & $\I \models_\D (a,e_d) : P \; \Longleftrightarrow \;   \langle a^\I , e_d^\D \rangle \in P^\I$ \\

\hline \caption{Semantics of $\shdlssx$.}\\
\caption*{\emph{Legenda.} ab: abstract, cn.: concrete, rl.: role, ind.: individual, d. cs.: data type constant, dtype: data type, ng.: negated, bot.: bottom, incl.: inclusion, asser.: assertion.}  \label{semdlss}
\end{longtable}}


Let $\mathcal{R}$, $\mathcal{T}$, and $\mathcal{A}$  be as above. An interpretation $\I= (\Delta ^ \I, \Delta_{\D}, \cdot ^ \I)$ is a $\D$-model of $\mathcal{R}$ (resp., $\mathcal{T}$), and we write $\I \models_{\D} \mathcal{R}$ (resp., $\I \models_{\D} \mathcal{T}$), if $\I$ satisfies each axiom in $\mathcal{R}$ (resp., $\mathcal{T}$) according to the semantic rules in Table \ref{semdlss}.  Analogously,  $\I= (\Delta^ \I, \Delta_{\D}, \cdot^\I)$ is a $\D$-model of $\mathcal{A}$, and we write $\I \models_{\D} \mathcal{A}$, if $\I$ satisfies each assertion in $\mathcal{A}$, according to the semantic rules in Table \ref{semdlss}. 

A $\shdlssx$-knowledge base $\mathcal{K}=(\mathcal{A}, \mathcal{T}, \mathcal{R})$ is consistent if there is an interpretation $\I= (\Delta^ \I, \Delta_{\D}, \cdot^\I)$ that is a $\D$-model of $\mathcal{A}$,  $\mathcal{T}$, and $\mathcal{R}$.

Decidability of the consistency problem for $\shdlssx$-knowledge bases was proved in \cite{CanLonNicSanRR2015} via a reduction to the satisfiability problem for formulae of
a four level quantified syllogistic called \flqsr. The latter problem was proved decidable in \cite{CanNic2013}.
%
%
Some considerations on the expressive power of $\shdlssx$ are in order. As illustrated in \cite[Table~ 1]{RR2017ext} existential quantification is admitted only on the left hand side of inclusion axioms. Thus $\shdlssx$ is less powerful than logics such as $\sroiqd\space$ \cite{Horrocks2006}
for what concerns the generation of new individuals. On the other hand, $\shdlssx$ is more liberal than $\sroiqd\space$ in the definition of role inclusion axioms since roles involved are not required to be subject to any ordering relationship, and the notion of simple role is not needed.
For example, the role hierarchy presented in \cite[page~ 2]{Horrocks2006} is not expressible in $\sroiqd\space$ but can be represented in $\shdlssx$. In addition, $\shdlssx$ is a powerful rule language able to express rules with  negated atoms such as  $Person(?p) \wedge \neg hasCar(?p, ?c) \implies CarlessPerson(?p)$. Notice that rules with negated atoms are not supported by the SWRL language.

\section{ ABox Reasoning services for $\shdlssx$ knowledge base}

The most important feature of a knowledge representation system is the capability of providing reasoning services. Depending on the type of the application domains, there are many different kinds of implicit knowledge that is desirable to infer from what is explicitly mentioned in the knowledge base.
In particular, reasoning problems regarding ABoxes consist in querying a knowledge base in order to retrieve information concerning data stored in it. 
In this section we study the decidability for the most widespread ABox reasoning tasks for the logic $\shdlssx$ resorting to a general problem, called Higher Order Conjuctive Query Answering (HOCQA), that can be instantiated to each of them. 


%
%
Let $\varind  = \{\sfvar{v}{1}, \sfvar{v}{2}, \ldots\}$,  $\varcon = \{\sfvar{c}{1}, \sfvar{c}{2}, \ldots\}$, $\varar = \{\sfvar{r}{1}, \sfvar{r}{2}, \ldots\}$, and $\varcr  = \{\sfvar{p}{1}, \sfvar{p}{2}, \ldots\}$ be pairwise disjoint denumerably infinite sets of variables which are  disjoint from $\Ind$, $\bigcup\{N_C(d): d \in N_{\D}\}$, $\C$, $\Ra$, and $\Rd$. A  HO $\shdlssx$-\emph{atomic formula} is an expression of one of the following types:
  $R(w_1,w_2)$, 
  $P(w_1, u_1)$, 
  $C(w_1)$, 
  $\mathsf{r}(w_1,w_2)$,
  $\mathsf{p}(w_1, u_1)$, 
  $\mathsf{c}(w_1)$, 
  $w_1=w_2$, 
  $u_1 = u_2$, 
where $w_1,w_2 \in \varind \cup \Ind$, $u_1, u_2 \in  \varind \cup \bigcup \{N_C(d): d \in N_{\D}\}$, $R$ is a $\shdlssx$-abstract role term, $P$ is a $\shdlssx$-concrete role term, $C$ is a $\shdlssx$-concept term, $\mathsf{r} \in \varar$, $\mathsf{p} \in \varcr$, and $\mathsf{c} \in \varcon$. A HO $\shdlssx$-atomic formula containing no variables is said to be \emph{ground}. A HO $\shdlssx$-\emph{literal} is a HO $\shdlssx$-atomic formula or its negation. 
 A HO $\shdlssx$-\emph{conjunctive query} is a conjunction of HO $\shdlssx$-literals. 
 We denote with $\lambda$ the \emph{empty} HO $\shdlssx$-conjunctive query.

Let  $\sfvar{v}{1},\ldots,\sfvar{v}{n} \in \varind$, $\sfvar{c}{1}, \ldots, \sfvar{c}{m} \in \varcon$, $\sfvar{r}{1}, \ldots, \sfvar{r}{k} \in \varar$, $\sfvar{p}{1}, \ldots, \sfvar{p}{h} \in \varcr$, $o_1, \ldots, o_n \in \Ind \cup \bigcup \{N_C(d): d \in N_{\D}\}$, $C_1, \ldots, C_m \in \C$, $R_1, \ldots, R_k \in \Ra$, and $P_1, \ldots, P_h \in \Rd$.
 A substitution\\
\centerline{ 
$\sigma  \defAs \{\sfvar{v}{1}/o_1, \ldots, \sfvar{v}{n}/o_n, \sfvar{c}{1}/{C_1}, \ldots, \sfvar{c}{m}/{C_m}, \sfvar{r}{1}/{R_1}, \ldots, \sfvar{r}{k}/{R_k}, \sfvar{p}{1} /{P_1}, \ldots, \sfvar{p}{h}/{P_h} \}
$}
is a map such that, for every  HO $\shdlssx$-literal $L$, $L\sigma$ is obtained from $L$ by replacing
the occurrences of $\sfvar{v}{i}$ in $L$ with $o_i$, for $i=1, \ldots, n$;
the occurrences of $\sfvar{c}{j}$ in $L$ with $C_j$, for $j=1, \ldots, m$;
the occurrences of $\sfvar{r}{\ell}$ in $L$ with $R_\ell$, for $\ell=1, \ldots, k$;
the occurrences of $\sfvar{p}{t}$ in $L$ with $P_t$, for $t=1, \ldots, h$.

Substitutions can be extended to HO $\shdlssx$-conjunctive queries in the usual way. 
Let $Q \defAs  (L_1 \wedge \ldots \wedge L_m)$ be a HO $\shdlssx$-conjunctive query, and $\KB$ a $\shdlssx$-knowledge base. A substitution $\sigma$ involving \emph{exactly} the variables occurring in $Q$ is a \emph{solution for $Q$ w.r.t. $\KB$} if there exists a $\shdlssx$-interpretation $\I$ such that $\I \models_{\D} \KB$ and $\I \models_{\D} Q \sigma$. The collection $\Sigma$ of the  solutions for $Q$ w.r.t. $\KB$ is the \emph{higher order (HO) answer set of $Q$ w.r.t. $\KB$}. Then the \emph{higher order conjunctive query answering} (HOCQA) problem for $Q$ w.r.t. $\KB$ consists in finding the HO answer set $\Sigma$ of $Q$ w.r.t. $\KB$. 
We shall solve the HOCQA problem just stated by reducing it to the analogous problem formulated in the context of the fragment $\flqsr$ (and in turn to the decision procedure for $\flqsr$ presented in \cite{CanNic2013}). The HOCQA problem for $\flqsr$-formulae can be stated as follows.
Let $\phi$ be a $\flqsr$-formula and let $\psi$ be a conjunction of $\flqsr$-quantifier-free atomic formulae of level $0$ of the types $x=y$, $x \in X^1$, $ \langle x,y \rangle \in X^3$,
or their negations.

The \emph{HOCQA problem for $\psi$ w.r.t.\ $\phi$} consists in computing the HO \emph{answer set of $\psi$ w.r.t.\ $\phi$}, namely the collection $\Sigma'$ of all the  substitutions $\sigma'$ such that  $\M \models \phi \wedge \psi\sigma'$, for some $\flqsr$-interpretation $\M$.

 
In view of the decidability of the satisfiability problem for $\flqsr$-formulae, the HOCQA problem for $\flqsr$-formulae is decidable as well. Indeed, let  $\phi$ and $\psi$ be two $\flqsr$-formulae fulfilling the above requirements. To calculate the HO answer set of $\psi$ w.r.t.\ $\phi$, for each candidate substitution\\[0.2cm] 
\centerline{$\sigma' \defAs \{ \vec{x} / \vec{z}, \vec{X^1} / \vec{Y^1}, \vec{X^2} / \vec{Y^2}, \vec{X^3} / \vec{Y^3}  \}$}\\[0.2cm] 
one has just to check for satisfiability of the $\flqsr$-formula $\phi \wedge \psi\sigma'$. Since the number of possible candidate substitutions is $\left| \Vars(\phi) \right|^{\left| \Vars(\psi) \right|}$
and the satisfiability problem for $\flqsr$-formulae is decidable, the HO answer set of $\psi$ w.r.t.\ $\phi$ can be computed effectively. Summarizing,
\begin{lemma}\label{CQA4LQSR}
The HOCQA problem for $\flqsr$-formulae is decidable. \qed
\end{lemma}

The following theorem states decidability of the HOCQA problem for $\shdlssx$.
\begin{theorem}\label{CQADL}
Given a $\shdlssx$-knowledge base $\KB$ and a HO $\shdlssx$- conjunc\-tive query $Q$, the HOCQA problem for $Q$ w.r.t. $\KB$ is decidable.
\end{theorem}
\begin{proof}
We first outline the main ideas and then we provide a formal proof of the theorem.

In order to define a $\flqsr formula$ $\phi_\KB$, we recall the definition a function $\theta$ that maps the $\shdlssx$-knowledge base $\KB$  in the $\flqsr$-formula in Conjunctive Normal Form (CNF) $\phi_{\KB}$, introduced in \cite{ExtendedVersionICTCS2016}.
The definition of the mapping $\theta$ is inspired to the definition of the mapping $\tau$ introduced in the proof of Theorem 1 in \cite{CanLonNicSanRR2015}. Specifically, $\theta$ differs from $\tau$ because it allows quantification only on variables of level $0$, it treats Boolean operations on concrete roles and the product of concepts, it constructs $\flqsr$-formulae in CNF and it is extended to $\shdlssx$-HO conjunctive queries.
 To prepare for the definition of $\theta$, we map injectively individuals $a\in \Ind$ and constants $e_d \in N_{C}(d)$
 into level $0$ variables $x_a$, $x_{e_d}$, the constant concepts $\top$ and $\bot$, data type terms $t$, and concept terms $C$ into level $1$ variables $X_{\top}^1$, $X_{\bot}^1$, $X_{t}^1$, $X_{C}^1$, respectively, and the universal relation on individuals $U$, abstract role terms $R$, and concrete role terms $P$ into level $3$ variables $X_{U}^3$, $X_{R}^3$, and $X_{P}^3$, respectively.\footnote{The use of level $3$ variables to model abstract and concrete role terms is motivated by the fact that their elements, that is ordered pairs $\langle x, y \rangle$, are encoded in Kuratowski's style as $\{\{x\}, \{x,y\}\}$, namely as collections of sets of objects.}

Then the mapping $\theta$ is defined as follows:
\smallskip

\noindent $\theta(C_1 \equiv \top) \defAs (\forall z)( ( \neg(z \in X_{C_1}^1) \vee z \in X_{\top}^1) \wedge ( \neg(z \in X_{\top}^1) \vee z \in X_{C_1}^1))$,

\noindent $\theta(C_1 \equiv \neg C_2) \defAs (\forall z)(( \neg(z \in X_{C_1}^1) \vee \neg(z \in X_{C_2}^1)) \wedge (z \in X_{C_2}^1 \vee z \in X_{C_1}^1))$,

\noindent $\theta(C_1 \equiv C_2 \sqcup C_3 ) \defAs (\forall z)( ( \neg(z \in X_{C_1}^1) \vee (z \in X_{C_2}^1 \vee z \in X_{C_3}^1)) \wedge ( (\neg (z \in X_{C_2}^1) \vee z \in X_{C_1}^1) \wedge (\neg (z \in X_{C_3}^1) \vee z \in X_{C_1}^1 ))$,

\noindent $\theta(C_1 \equiv \{a\}) \defAs (\forall z)( \neg(z \in X_{C_1}^1) \vee z = x_a) \wedge( \neg(z = x_a) \vee z \in X_{C_1}^1 )$,

\noindent $\theta(C_1 \sqsubseteq \forall R_1.C_2) \defAs (\forall z_1)(\forall z_2)( \neg(z_1 \in X_{C_1}^1) \vee ( \neg(\langle z_1,z_2 \rangle \in X_{R_1}^3) \vee z_2 \in X_{C_2}^1))$,

\noindent $\theta(\exists R_1.C_1 \sqsubseteq C_2) \defAs (\forall z_1)(\forall z_2)(( \neg(\langle z_1,z_2 \rangle \in X_{R_1}^3) \vee \neg( z_2 \in X_{C_1}^1)) \vee z_1 \in X_{C_2}^1)$,

\noindent $\theta(C_1 \equiv \exists R_1.\{a\}) \defAs(\forall z)( ( \neg(z \in X_{C_1}^1) \vee \langle z,x_{a}\rangle \in X_{R_1}^3) \wedge ( \neg(\langle z,x_{a}\rangle \in X_{R_1}^3) \vee z \in X_{C_1}^1  ) )$,

\noindent $\theta(C_1 \sqsubseteq \leq_n\!\! R_1.C_2) \defAs (\forall z)(\forall z_1)\ldots (\forall z_{n+1})( \neg(z \in X_{C_1}^1) \vee  ( \overset{n+1}{\underset{i=1}\bigwedge}( \neg(z_i \in X_{C_2}) \vee \neg(\langle z,z_i\rangle \in X_{R_1}^3) \vee \underset {i<j} {\bigvee} z_i = z_j))$,

\noindent $\theta(\geq_n\!\! R_1.C_1 \sqsubseteq C_2) \defAs (\forall z)(\forall z_1)\ldots (\forall z_{n})( \overset {n}{ \underset{i=1}\bigwedge }(( \neg(z_i \in X_{C_1}^1) \vee \neg( \langle z,z_i\rangle \in X_{R_1}^3)) \vee  \underset {i<j} \bigvee z_i = z_j) \vee z \in X_{C_2}^1)$,

\noindent $\theta(C_1 \sqsubseteq \forall P_1.t_1) \defAs (\forall z_1)(\forall z_2)( \neg(z_1 \in X_{C_1}^1) \vee ( \neg (\langle z_1,z_2 \rangle \in X_{P_1}^3) \vee z_2 \in X_{t_1}^1))$,

\noindent $\theta(\exists P_1.t_1 \sqsubseteq C_1) \defAs (\forall z_1)(\forall z_2)(( \neg(\langle z_1,z_2 \rangle \in X_{P_1}^3) \vee \neg(z_2 \in X_{t_1}^1)) \vee z_1 \in X_{C_1}^1)$,

\noindent $\theta(C_1 \equiv \exists P_1.\{e_{d}\}) \defAs (\forall z)( ( \neg(z \in X_{C_1}^1) \vee \langle z,x_{e_{d}}\rangle \in X_{P_1}^3)  \wedge ( \neg(\langle z,x_{e_{d}}\rangle \in X_{P_1}^3) \vee z \in X_{C_1}^1) )$,

\noindent $\theta(C_1 \sqsubseteq \leq_n\!\! P_1.t_1) \defAs (\forall z)(\forall z_1)\ldots (\forall z_{n+1})( \neg (z \in X_{C_1}^1) \vee ( \overset {n+1} { \underset{i=1}\bigwedge }( \neg(z_i \in X_{t_1}) \vee \neg(\langle z,z_i\rangle \in X_{P_1}^3) \vee \underset{i<j} {\bigvee} z_i = z_j))$,

\noindent $\theta(\geq_n\!\! P_1.t_1 \sqsubseteq C_1) \defAs (\forall z)(\forall z_1)\ldots (\forall z_{n})( \overset {n} { \underset {i=1}\bigwedge}(( \neg(z_i \in X_{t_1}^1) \vee \neg(\langle z,z_i\rangle \in X_{P_1}^3)) \vee \underset {i<j} {\bigvee} z_i = z_j) \vee z \in X_{C_1}^1)$,

\noindent $\theta(R_1 \equiv U) \defAs (\forall z_1)(\forall z_2)( ( \neg(\langle z_1,z_2\rangle \in X_{R_1}^3) \vee \langle z_1,z_2\rangle \in X_{U}^3) \wedge ( \neg(\langle z_1,z_2\rangle \in X_{U}^3) \vee \langle z_1,z_2\rangle \in X_{R_1}^3) )$,

\noindent $\theta(R_1 \equiv \neg R_2) \defAs (\forall z_1)(\forall z_2)( ( \neg(\langle z_1,z_2\rangle \in X_{R_1}^3) \vee \neg (\langle z_1,z_2\rangle \in X_{R_2}^3 )) \wedge ( \langle z_1,z_2\rangle \in X_{R_2}^3 \vee \neg (\langle z_1,z_2\rangle \in X_{R_1}^3 )) )$,

\noindent $\theta( R \equiv C_1 \times C_2 ) \defAs (\forall z_1)(\forall z_2) (  \neg (\langle z_1, z_2 \rangle \in X^3_R) \vee  z_1 \in X^1_{C_1}) \wedge ( \neg (\langle z_1, z_2 \rangle \in X^3_R) \vee  z_2  \in X^1_{C_2}  ) \wedge  (( \neg(z_1 \in X^1_{C_1}) \vee  \neg(z_2  \in X^1_{C_2})  ) \vee \langle z_1, z_2 \rangle \in X^3_R  ) )$

\noindent $\theta(R_1 \equiv R_2 \sqcup R_3) \defAs (\forall z_1)(\forall z_2)( ( \neg(\langle z_1,z_2 \rangle \in X_{R_1}^3) \vee (\langle z_1,z_2 \rangle \in X_{R_2}^3 \vee \langle z_1,z_2 \rangle \in X_{R_3}^3)) \wedge ( ( \neg(\langle z_1,z_2 \rangle \in X_{R_2}^3) \vee \langle z_1,z_2 \rangle \in X_{R_1}^3) \wedge ((\neg(\langle z_1,z_2 \rangle \in X_{R_3}^3) \vee \langle z_1,z_2 \rangle \in X_{R_1}^3) ) ))$,

\noindent $\theta(R_1 \equiv R_2^{-}) \defAs (\forall z_1)(\forall z_2)( ( \neg(\langle z_1,z_2\rangle \in X_{R_1}^3) \vee \langle z_2,z_1\rangle \in X_{R_2}^3 ) \wedge ( \neg(\langle z_2,z_1\rangle \in X_{R_2}^3) \vee \langle z_1,z_2\rangle \in X_{R_1}^3  )   )  $,

\noindent $\theta(R_1 \equiv id(C_1)) \defAs (\forall z_1)(\forall z_2)( ( ( \neg(\langle z_1,z_2\rangle \in X_{R_1}^3) \vee z_1 \in X_{C_1}^1 ) \wedge ( \neg(\langle z_1,z_2\rangle \in X_{R_1}^3) \vee z_2 \in X_{C_1}^1 ) \wedge (\neg(\langle z_1,z_2\rangle \in X_{R_1}^3) \vee z_1 =z_2)  )\wedge ( ( \neg(z_1 \in X_{C_1}^1) \vee \neg(z_2 \in X_{C_1}^1) \vee z_1 \neq z_2) \vee \langle z_1,z_2\rangle \in X_{R_1}^3)  )$,

\noindent $\theta(R_1 \equiv R_{2_{C_1 |}}) \defAs (\forall z_1)(\forall z_2)( ( (\neg(\langle z_1,z_2\rangle \in X_{R_1}^3) \vee \langle z_1,z_2\rangle \in X_{R_2}^3) \wedge  ( \neg(\langle z_1,z_2\rangle \in X_{R_1}^3) \vee z_1 \in X_{C_1}^1))  \wedge (( \neg(\langle z_1,z_2\rangle \in X_{R_2}^3) \vee \neg(z_1 \in X_{C_1}^1))  \vee \langle z_1,z_2\rangle \in X_{R_1}^3 ) )$,

\noindent $\theta(R_1 \ldots R_n \sqsubseteq R_{n+1}) \defAs (\forall z)(\forall z_1)\ldots (\forall z_{n}) (( \neg(\langle z, z_1\rangle \in X_{R_1}^3) \vee \ldots \vee \neg(\langle z_{n-1},z_n\rangle \in X_{R_{n}}^3)) \vee  \langle z, z_n\rangle\in X_{R_{n+1}}^3)$,

\noindent $\theta(\refl(R_1)) \defAs (\forall z)(\langle z, z\rangle \in X_{R_1}^3)$,

\noindent $\theta(\irref(R_1)) \defAs (\forall z)(\neg (\langle z,z\rangle \in X_{R_1}^3))$,

\noindent $\theta(\fun(R_1)) \defAs (\forall z_1)(\forall z_2)(\forall z_3)(( \neg(\langle z_1,z_2\rangle \in X_{R_1}^3) \vee \neg(\langle z_1,z_3\rangle \in X_{R_1}^3)) \vee z_2 =z_3)$,

\noindent $\theta(P_1 \equiv P_2) \defAs (\forall z_1)(\forall z_2)( ( \neg (\langle z_1,z_2\rangle \in X_{P_1}^3) \vee \langle z_1,z_2\rangle \in X_{P_2}^3) \wedge ( \neg(\langle z_1,z_2\rangle \in X_{P_2}^3) \vee  \langle z_1,z_2\rangle \in X_{P_1}^3)  )$,

\noindent $\theta(P_1 \equiv \neg P_2) \defAs (\forall z_1)(\forall z_2)( ( \neg(\langle z_1,z_2\rangle \in X_{P_1}^3) \vee \neg(\langle z_1,z_2\rangle \in X_{P_2}^3)) \wedge (\langle z_1,z_2\rangle \in X_{P_2}^3 \vee \langle z_1,z_2\rangle \in X_{P_1}^3 ) )$,

\noindent $\theta(P_1 \sqsubseteq P_2) \defAs (\forall z_1)(\forall z_2)( \neg(\langle z_1,z_2\rangle \in X_{P_1}^3) \vee \langle z_1,z_2\rangle \in X_{P_2}^3)$,

\noindent $\theta(\fun(P_1)) \defAs (\forall z_1)(\forall z_2)(\forall z_3)( ( \neg(\langle z_1,z_2\rangle \in X_{P_1}^3) \vee \neg(\langle z_1,z_3\rangle \in X_{P_1}^3) \vee z_2 =z_3)$,

\noindent $\theta(P_1 \equiv P_{2_{C_1 |}}) \defAs (\forall z_1)(\forall z_2) (( \neg(\langle z_1,z_2\rangle \in X_{P_1}^3) \vee \langle z_1,z_2\rangle \in X_{P_2}^3 ) \wedge ( \neg(\langle z_1,z_2\rangle \in X_{P_1}^3) \vee z_1 \in X_{C_1}^1) \wedge ( (\neg \langle z_1,z_2\rangle \in X_{P_2}^3) \vee \neg(z_1 \in X_{C_1}^1) \vee  \langle z_1,z_2\rangle \in X_{P_1}^3) )$,

\noindent $\theta(P_1 \equiv P_{2_{|t_1}}) \defAs (\forall z_1)(\forall z_2) ( ( \neg(\langle z_1,z_2\rangle \in X_{P_1}^3) \vee \langle z_1,z_2\rangle \in X_{P_2}^3 ) \wedge ( \neg(\langle z_1,z_2\rangle \in X_{P_1}^3) \vee z_2 \in X_{t_1}^1) \wedge ( ( \neg(\langle z_1,z_2\rangle \in X_{P_2}^3) \vee \neg(z_2 \in X_{t_1}^1)) \vee \langle z_1,z_2\rangle \in X_{P_1}^3 ) )$,

\noindent $\theta(P_1 \equiv P_{2_{C_1|t_1}}) \defAs (\forall z_1)(\forall z_2) ( ( \neg(\langle z_1,z_2\rangle \in X_{P_1}^3) \vee \langle z_1,z_2\rangle \in X_{P_2}^3 ) \wedge ( \neg(\langle z_1,z_2\rangle \in X_{P_1}^3) \vee z_1 \in X_{C_1}^1) \wedge ( \neg(\langle z_1,z_2\rangle \in X_{P_1}^3) \vee z_2 \in X_{t_1}^1)
 \wedge (    \neg(\langle z_1,z_2\rangle \in X_{P_2}^3) \vee \neg(z_1 \in X_{C_1}^1) \vee \neg(z_2 \in X_{t_1}^1)  \vee \langle z_1,z_2\rangle \in X_{P_1}^3) )$,

\noindent $\theta(t_1 \equiv t_2)\defAs (\forall z)(( \neg(z \in X_{t_1}^1) \vee z \in X_{t_2}^1) \wedge ( \neg(z \in X_{t_2}^1) \vee z \in X_{t_1}^1 ))$,
\noindent $\theta(t_1 \equiv \neg t_2)\defAs (\forall z)(( \neg(z \in X_{t_1}^1) \vee \neg (z \in X_{t_2}^1)) \wedge ( z \in X_{t_2}^1  \vee z \in X_{t_1}^1) )$,

\noindent $\theta(t_1 \equiv t_2 \sqcup t_3)\defAs (\forall z)( ( \neg(z \in X_{t_1}^1) \vee (z \in X_{t_2}^1\vee z \in X_{t_3}^1)) \wedge ( ( \neg (z \in X_{t_2}^1) \vee  z \in X_{t_1}^1) \wedge (\neg (z \in X_{t_3}^1) \vee  z \in X_{t_1}^1) ) )$,

\noindent $\theta(t_1 \equiv t_2 \sqcap t_3)\defAs (\forall z)(( \neg(z \in X_{t_1}^1) \vee (z \in X_{t_2}^1 \wedge z \in X_{t_3}^1))  \wedge ((( \neg(z \in X_{t_2}^1) \vee \neg(z \in X_{t_3}^1))  \vee z \in X_{t_1}^1)  )$,

\noindent $\theta(t_1 \equiv \{e_{d}\})\defAs (\forall z)((\neg (z \in X_{t_1}^1) \vee z = x_{e_{d}}) \wedge ( \neg(z = x_{e_{d}})  \vee z \in X_{t_1}^1) )$,

\noindent $\theta(a : C_1) \defAs x_a \in X_{C_1}^1$,

\noindent $\theta((a,b) : R_1) \defAs \langle x_a, x_b\rangle \in X_{R_1}^3$,

\noindent $\theta((a,b) : \neg R_1) \defAs \neg(\langle x_a, x_b\rangle \in X_{R_1}^3)$,

\noindent $\theta(a=b) \defAs x_a = x_b$, $\theta(a\neq b) \defAs \neg (x_a = x_b)$,

\noindent $\theta(e_d : t_1) \defAs x_{e_d} \in X_{t_1}^1$,

\noindent $\theta((a,e_d) : P_1) \defAs \langle x_a, x_{e_d}\rangle \in X_{P_1}^3$, $\theta((a,e_d) : \neg P_1) \defAs \neg(\langle x_a, x_{e_d}\rangle \in X_{P_1}^3)$,

\noindent $\theta(\alpha \wedge \beta) \defAs \theta(\alpha) \wedge \theta(\beta)$.

%

Let $\mathcal{KB}$ be our $\shdlssx$-knowledge base, and let $\ck$, $\ark$, $\crk$, and $\ik$ be, respectively, the sets of concept, of abstract role, of concrete role, and of individual names in $\mathcal{KB}$. Moreover, let $N_{D}^\mathcal{KB} \subseteq N_{D}$ be the set of data types in $\mathcal{KB}$, $N_{F}^\mathcal{KB}$ a restriction of $N_{F}$ assigning to every $d \in N_{\D}^\mathcal{KB}$ the set $N_{F}^\mathcal{KB}(d)$ of facets in $N_{F}(d)$ and in $\mathcal{KB}$. Analogously, let $N_{C}^{\mathcal{KB}}$ be a restriction of the function $N_{C}$ associating to every $d \in N_{\D}^\mathcal{KB}$ the set $N_{C}^\mathcal{KB}(d)$ of constants  contained in $N_{C}(d)$ and in $\mathcal{KB}$. Finally, for every data type $d \in N_{D}^\mathcal{KB}$, let $\bfk(d)$ be the set of facet expressions for $d$ occurring in $\mathcal{KB}$ and not in $N_{F}(d) \cup \{\top^{d},\bot_{d}\}$. We assume without loss of generality that the facet expressions in $\bfk(d)$ are in Conjunctive Normal Form. We define the \flqsr-formula $\phi_{\mathcal{KB}}$ expressing the consistency of $\mathcal{KB}$ as follows: {\small
\[
\phi_{\mathcal{KB}} \defAs \underset {H \in \mathcal{KB}}\bigwedge \theta(H) \wedge \bigwedge_{i=1}^{12}\xi_i   \, ,
\]}
where

\noindent$\xi_1 \defAs (\forall z)( ( \neg(z \in X_{\I}^1) \vee \neg(z \in X_{\D}^1)) \wedge (z \in X_{\D}^1 \vee  z \in X_{\I}^1))\wedge (\forall z)(z \in $

\noindent $\hfill  X_{\I}^1 \vee z \in X_{\D}^1)\wedge \neg (\forall z)\neg (z \in X_{\I}^1) \wedge \neg (\forall z)\neg (z \in X_{\D}^1)$,\\
 

\noindent$\xi_2 \defAs ((\forall z)( ( \neg(z \in X_{\I}^1) \vee z \in X_{\top}^1) \wedge (\neg (z \in X_{\top}^1)  \vee z \in X_{\I}^1) ) \wedge (\forall z)\neg (z \in$

\noindent $\hfill  X_{\bot})$,\\


\noindent$\xi_3 \defAs \underset{A \in \ck}\bigwedge (\forall z)( \neg(z \in X_{A}^1) \vee z \in X_{\I}^1)$,\\


 \noindent$\xi_4 \defAs ( \underset{d \in N_{D}^\mathcal{KB}}\bigwedge((\forall z)( \neg(z \in X_{d}^1) \vee z \in X_{\D}^1) \wedge \neg (\forall z)\neg(z \in X_{d}^1))  \wedge (\forall z)$
 
 \noindent $\hfill  (\underset{(d_i,d_j \in N_{D}^\mathcal{KB}, i < j)}\bigwedge ( ( \neg(z \in X_{d_i}^1) \vee \neg (z \in X_{d_j}^1)) \wedge ( z \in X_{d_j}^1 \vee  z \in X_{d_i}^1 ) )))$,\\
 

\noindent$\xi_5 \defAs \underset{d \in N_{D}^\mathcal{KB}}\bigwedge((\forall z)( ( \neg(z \in X_{d}^1) \vee z \in X_{\top_d}^1) \wedge ( \neg(z \in X_{\top_d}^1)  \vee z \in X_{d}^1  ) \wedge $

\noindent $\hfill (\forall z)\neg(z \in X_{\bot_d}^1))$,\\


\noindent$\xi_6 \defAs   \underset {\substack{\\ f_d \in N_{F}^\mathcal{KB}(d),\\ d \in N_{D}^\mathcal{KB}}}\bigwedge (\forall z)( \neg(z \in X_{f_d}^1) \vee z \in X_{d}^1)$,\\


\noindent$\xi_7 \defAs (\forall z_1)(\forall z_2)( ( \neg(z_1 \in X_{\I}^1) \vee \neg(z_2 \in X_{\I}^1) \vee \langle z_1,z_2 \rangle \in X_{U}^3) \wedge (    (\neg(\langle z_1,z_2 \rangle \in$

\noindent $\hfill  X_{U}^3) \vee z_1 \in X_{\I}^1) \wedge ( \neg(\langle z_1,z_2 \rangle \in X_{U}^3) \vee z_2 \in X_{\I}^1 ) ) )$,\\


\noindent$\xi_8 \defAs  \underset {R \in \ark} \bigwedge
(\forall z_1)(\forall z_2) ( (\neg(\langle z_1,z_2\rangle \in X_R^3) \vee z_1 \in X_{\I}^1 ) \wedge ( \neg(\langle z_1,z_2\rangle \in X_R^3) \vee z_2 \in X_{\I}^1)))$,\\


\noindent$\xi_9 \defAs \underset{T \in \crk} \bigwedge  (\forall z_1)(\forall z_2) (\neg(\langle z_1,z_2\rangle \in X_T^3) \vee z_1 \in X_{\I}^1) \wedge ( \neg(\langle z_1,z_2\rangle \in X_T^3) \vee z_2 \in X_{\D}^1)))$,\\


\noindent$\xi_{10} \defAs \underset{a \in \ik} \bigwedge(x_a \in X_{\I}^1) \wedge \underset { \substack{ \\ d \in N_{D}^\mathcal{KB}, \\{e_d \in N_{C}^\mathcal{KB}(d)}}} \bigwedge  x_{e_d} \in X_{d}^1$,\\


\noindent$\xi_{11} \defAs \underset {\{e_{d_1},\ldots, e_{d_n}\} \textrm{ in } \mathcal{KB}} \bigwedge (\forall z) ( ( \neg(z \in X_{\{e_{d_1},\ldots, e_{d_n}\}}^1) \vee \overset{n}  { \underset {i=1} \bigvee }(z = x_{e_{d_i}})) \wedge (  \overset{n}  { \underset {i=1} \bigwedge }(z \neq$

\noindent $\hfill  x_{e_{d_i}}  \vee  z \in X_{\{e_{d_1},\ldots, e_{d_n}\}}^1 ) ) )
 \wedge \quad \underset {\{a_{1},\ldots, a_{n}\} \textrm{ in } \mathcal{KB}} \bigwedge (\forall z)( (\neg(z \in X_{\{a_{1},\ldots, a_{n}\}}^1) \vee  $
 
 \noindent $\hfill \overset {n}{\underset {i=1} \bigvee}(z = x_{a_{i}})) \wedge (  \overset {n}{\underset {i=1} \bigwedge}(z \neq x_{a_{i}} \vee z \in X_{\{a_{1},\ldots, a_{n}\}}^1)) )$,\\
 

\noindent$\xi_{12} \defAs \underset { \substack{d \in N_{\D}^\mathcal{KB},\\[1.5pt] {\psi_d \in \bfk(d)}}} \bigwedge  (\forall z) ( \neg(z \in X_{\psi_d}^1) \vee z \in \zeta(X_{\psi_d}^1)) \wedge ( \neg(z \in \zeta(X_{\psi_d}^1)) \vee z \in X_{\psi_d}^1 )$


\medskip

with $\zeta$ the transformation function from \flqsr-variables of level 1 to \flqsr-formulae recursively defined, for $d \in N_\D^\mathcal{KB}$, by
\[ {
\zeta(X_{\psi_d}^1) \defAs \begin{cases}
X_{\psi_d}^1 & \text{if } \psi_d \in N_{F}^\mathcal{KB}(d) \cup \{\top^{d},\bot_{d}\}\\
\neg \zeta(X_{\chi_d}^1) & \text{if }  \psi_d = \neg \chi_d\\
\zeta(X_{\chi_d}^1) \wedge \zeta(X_{\varphi_d}^1) & \text{if } \psi_d = \chi_d \wedge \varphi_d\\
\zeta(X_{\chi_d}^1) \vee \zeta(X_{\varphi_d}^1) & \text{if } \psi_d = \chi_d \vee \varphi_d\,.
\end{cases} }
\]
\noindent In the above formulae, the variable $X_{\I}^1$ denotes the set of individuals $\Ind$, $X_{d}^1$ a data type $d \in N_{D}^\mathcal{KB}$, $X_{\D}^1$  a superset of the union of data types in $ N_{D}^\mathcal{KB}$, $X_{\top_d}^1$ and $X_{\bot_d}^1$ the constants $\top_d$ and $\bot_d$, and $X_{f_d}^1$, $X_{\psi_d}^1$ a facet $f_d$ and a facet expression $\psi_d$, for $d \in N_{D}^\mathcal{KB}$, respectively. In addition, $X_{A}^1$, $X_{R}^3$, $X_{T}^3$ denote a concept name $A$, an abstract role name $R$, and a concrete role name $T$ occurring in $\mathcal{KB}$, respectively. Finally, $X_{\{e_{d_1},\ldots,e_{d_n}\}}^1$ denotes a data range $\{e_{d_1},\ldots,e_{d_n}\}$ occurring in $\mathcal{KB}$, and $X_{\{a_{1},\ldots,a_{n}\}}^1$ a finite set $\{a_1,\ldots,a_n\}$ of nominals in $\mathcal{KB}$.

 The constraints $\xi_1-\xi_{12}$, slightly different from the constraints $\psi_1-\psi_{12}$ defined in the proof of Theorem 1 in \cite{CanLonNicSanRR2015}, are introduced to guarantee that each model of $\phi_{\KB}$ can be easily transformed in a $\shdlssx$-interpretation.

 The HOCQA problem for $\shdlssx$ can be solved via an effective reduction to the HOCQA problem for $\flqsr$-formulae, and then exploiting Lemma \ref{CQA4LQSR}. The reduction is accomplished through the function $\theta$ extended in order to map also $\shdlssx$-conjunctive queries into $\flqsr$-formulae in conjunctive normal form (CNF), which can be used to map effectively HOCQA problems from the $\shdlssx$-context into the $\flqsr$-context. More specifically, given a $\shdlssx$-knowledge base $\KB$ and a $\shdlssx$-HO conjunctive query $Q$, using the function $\theta$ we can effectively construct the following $\flqsr$-formulae in CNF:\\[0.2cm]
\centerline{
$\phi_{\KB} \defAs \bigwedge_{H \in \KB} \theta(H) \wedge \bigwedge_{i=1}^{12} \xi_i, \qquad \psi_Q \defAs \theta(Q)\,.$
}\\[.2cm]
Then, if we denote by $\Sigma$ the higher order answer set of $Q$ w.r.t.\ $\KB$ and by $\Sigma'$ the higher order answer set of $\psi_Q$ w.r.t.\ $\phi_{\KB}$, we have that $\Sigma$ consists of all substitutions $\sigma$ (involving exactly the variables occurring in $Q$) such that $\theta(\sigma) \in \Sigma'$.
Since, by Lemma~\ref{CQA4LQSR}, $\Sigma'$ can be computed effectively, then $\Sigma$ can be computed effectively too.

\medskip \noindent The mapping $\theta$  is extended for $\shdlssx$-HO conjuctive queries as follows.\\

\noindent$\theta (R_1 (w_1,w_2)) \defAs \langle x_{w_1}, x_{w_2} \rangle \in X^3_{R_1}$,\\
$\theta (P_1 (w_1,u_1)) \defAs \langle x_{w_1}, x_{u_1} \rangle \in X^3_{P_1}$,\\
$\theta ( C_1(w_1) \defAs x_{w_1} \in X^1_{C_1}$, \\
$\theta (w_1 =w_2) \defAs x_{w_1} = x_{w_2}$, \\
$\theta (u_1 =u_2) \defAs x_{u_1} = x_{u_2}$.\\
$\theta (\sfvar{c}{1} (w_1)) \defAs  w_1 \in X^1_{\sfvar{c}{1}}$.\\ 
$\theta (\sfvar{r}{1} (w_1,w_2)) \defAs \langle w_1, w_2 \rangle \in X^3_{\sfvar{r}{1}}$.\\  
$\theta (\sfvar{p}{1} (w_1,u_1)) \defAs \langle w_1, u_1 \rangle \in X^3_{\sfvar{p}{1}}$.\\ 



To complete, we extend the mapping $\theta$ on substitutions 
\begin{equation*}
\sigma \defAs \{ \sfvar{v}{1}/o_1, \ldots \sfvar{v}{n} / o_n, \sfvar{c}{1} / C_1, \ldots, \sfvar{c}{m} / C_m, \sfvar{r}{1} / R_1, \ldots, \sfvar{r}{k}, / R_k, \sfvar{p}{1} / P_1, \ldots \sfvar{p}{h} / P_h \}
\end{equation*}

with $\sfvar{v}{1},\ldots,\sfvar{v}{n} \in \varind$, $\sfvar{c}{1}, \ldots, \sfvar{c}{m} \in \varcon$, $\sfvar{r}{1}, \ldots, \sfvar{r}{k} \in \varar$, $\sfvar{p}{1}, \ldots, \sfvar{p}{h} \in \varcr$, $o_1, \ldots, o_n \in \Ind \cup \bigcup \{N_C(d): d \in N_{\D}\}$, $C_1, \ldots, C_m \in \C$, $R_1, \ldots, R_k \in \Ra$, and $P_1, \ldots, P_h \in \Rd$.
 
 We put  
 \begin{equation} \label{sigma1}
 \begin{split} 
  \theta(\sigma)= & \theta ( \{ \sfvar{v}{1}/o_1, \ldots \sfvar{v}{n} / o_n, \sfvar{c}{1} / C_1, \ldots, \sfvar{c}{m} / C_m, \sfvar{r}{1} / R_1, \ldots, \sfvar{r}{k}, / R_k,  \\ &
    \sfvar{p}{1} / P_1, \ldots \sfvar{p}{h} / P_h \}) \\ 
  = & \{ x_{\sfvar{v}{1}}/x_{o_1}, \ldots, x_\sfvar{v}{n} / x_{o_n}, X^1_\sfvar{c}{1} / X^1_{C_1}, \ldots, X^1_\sfvar{c}{m} / X^1_{C_m},  X^3_\sfvar{r}{1} / X^3_{R_1}, \ldots, X^3_\sfvar{r}{k}, / X^3_{R_k},\\
    & X^3_{\sfvar{p}{1}} / X^3_{P_1}, \ldots, X^3_\sfvar{p}{h} / X^3_{P_h} \} \\  = & \sigma' \\
 \end{split}
 \end{equation}

 

\noindent where $ x_{\sfvar{v}{1}}, \ldots x_\sfvar{v}{n}$, $x_{o_1}, \ldots, x_{o_n}$ are variables of level $0$, $X^1_\sfvar{c}{1}, \ldots, X^1_\sfvar{c}{m}$, $X^1_{C_1}, \ldots, X^1_{C_m}$ are variables of level $1$, $ X^3_\sfvar{r}{1}, \ldots, X^3_\sfvar{r}{k}$, $X^3_\sfvar{p}{1}, \ldots, X^3_\sfvar{p}{h}$, $X^3_{R_1}, \ldots, X^3_{R_k}$, and $X^3_{P_1}, \ldots, X^3_{P_h}$ are variables of level 3 in $\flqsr$.   
 


To prove the theorem,  we show that $\Sigma$ is the higher order answer set for $Q$ w.r.t. $\KB$ iff $\Sigma$ is equal to $\overset{}{\underset{\M \models \phi_{\KB}}{\bigcup}}  \Sigma_{\M}'$, where $\Sigma_{\M}'$ is the collection of substitutions $\sigma$ such that $\M \models \psi_{Q}\sigma$. 
Let us assume that $\Sigma$ is higher order the answer set for $Q$ w.r.t. $\KB$. We have to show that $\Sigma$ is equal to $ \Sigma' = \underset{\M \models \phi_{\KB}}{\bigcup} \Sigma'_{\M}$, where $\Sigma'_{\M}$ is the collection of all the substitutions $\sigma'$ such that $\M \models \psi_{Q}\sigma'$.
 
By contradiction, let us assume that there exists a $\sigma \in \Sigma$ such that $\sigma \notin \Sigma'$, namely $\M \not\models \psi_{Q}\sigma$, for every $\flqsr$-interpretation $\M$ with $\M \models \phi_{\KB}$. Since $\sigma \in \Sigma$ there is a $\shdlssx$-interpretation $\I$ such that $\I \models_{\D} \KB$ and $\I \models_{\D} Q\sigma$. Then, by the construction above, we can define a $\flqsr$-interpretation $\M_{\I}$ such that $\M_{\I} \models \phi_{\KB}$ and $\M_{\I}\models \psi_Q\theta{\sigma}$. Absurd. 

Conversely, let $\sigma' \in \Sigma'$ and assume by contradiction that $\sigma' \notin \Sigma$. 
Then, for all $\shdlssx$-interpretations such that $\I \models_D \KB$, it holds that $\I \not\models_D Q\sigma'$. 
Since $\sigma'\in \Sigma'$, there is a $\flqsr$-interpretation $\M$ such that $\M \models \phi_{\KB}$ and $\M \models \psi\sigma'$. Then, by the construction above, we can define a $\shdlss$-interpretation $\I_{\M}$ such that  $\I_{\M} \models_D \KB$ and $\I_{\M} \models_D Q\sigma'$. Absurd.  \qed
\end{proof}
In what follows we list the most widespread reasoning services for $\shdlssx$ -ABox and then show how to define them as particular cases of the HOCQA task. 

\begin{enumerate}
\item{\emph{Instance checking}: the problem of deciding whether or not an individual $a$ is an instance of a concept $C$.}
\item{\emph{Instance retrieval}: the problem of retrieving all the individuals that are instances of a given concept.}
\item{\emph{Role filler retrieval}: the problem of retrieving all the fillers $x$ such that the pair $(a,x)$ is an instance of a role $R$.}
\item{\emph{Concept retrieval}: the problem of retrieving all concepts which an individual is an instance of.}
\item{\emph{Role instance retrieval}: the problem of retrieving all roles which a pair of individuals $(a,b)$ is an instance of.}
\end{enumerate}
The instance checking problem is a specialization of the HOCQA problem admitting  HO $\shdlssx$-conjunctive queries of the form $Q_{IC}=C(w_1)$, with $w_1 \in \Ind$. The  instance retrieval problem is a particular case of the HOCQA problem in which HO $\shdlssx$-conjunctive queries have the form $Q_{IR}=C(w_1)$, where $w_1$ is a variable in $\varind$. The HOCQA problem can be instantiated to the role filler retrieval problem by admitting HO $\shdlssx$-conjunctive queries $Q_{RF}=R(w_1,w_2)$,  with $w_1 \in \Ind$ and $w_2$ a variable in $\varind$. The  concept retrieval problem is a specialization of the HOCQA problem allowing  HO $\shdlssx$-conjunctive queries of the form $Q_{QR}=c(w_1)$, with $w_1 \in \Ind$ and $c$ a  variable in $\varcon$. Finally, the role instance retrieval problem is a particularization of the HOCQA problem, where  HO $\shdlssx$-conjunctive queries have the form $Q_{RI}=r(w_1,w_2)$, with $w_1,w_2 \in \Ind$ and $r$  a  variable in $\varcr$.
%

Notice that the CQA problem for 
$\shdlssx$ defined in \cite{ictcs16} is an instance of the HOCQA problem admitting HO $\shdlssx$-conjunctive queries of the form $Q_{CQA} \defAs (L_1 \wedge \ldots \wedge L_m)$, with $L_i$ an atomic formula of any of the types $R(w_1,w_2)$, $C(w_1)$, and $w_1=w_2$ (or their negation), where  $w_1,w_2 \in(\Ind \cup\varind)$. Notice also that problems 1, 2, and 3 are instances of the CQA problem for $\shdlssx$, whereas problems 4 and 5
fall outside the definition of CQA. As shown above, they can be treated as specializations of HOCQA.  

\section{An algorithm for the HOCQA problem for $\shdlssx$ }
In this section we introduce an effective set-theoretic procedure to compute the answer set of a HO $\shdlssx$-conjunctive query $Q$ w.r.t. a $\shdlssx$ knowledge base $\KB$. Such procedure, called \textit{HOCQA-$\shdlssx$}, takes as input $\phi_\KB$ (i.e., the $\flqsr$-translation of $\KB$) and $\psi_Q$ (i.e., the $\flqsr$-formula representing the HO $\shdlssx$-conjunctive query $Q$), and returns a  \ke\space $\T_\KB$, representing the saturation of $\KB$, and the answer set $\Sigma'$ of $\psi_Q$ w.r.t. $\phi_\KB$, namely the collection of all substitutions $\sigma'$ such that $\M \models \phi_\KB \wedge \psi_Q\sigma'$, for some $\flqsr$-interpretation $\M$. 
Specifically, \textit{HOCQA}-$\shdlssx$ constructs for each open branch of $\T_\KB$ a decision tree whose leaves are labelled with elements of $\Sigma'$.

In the following we introduce definitions, notions, and notations useful for the presentation of Procedure \textit{HOCQA}-$\shdlssx$.

Assume without loss of generality that universal quantifiers in $\phi_\KB$ occur as inward as possible and that universally quantified variables are pairwise distinct.
Let $S_1, \ldots, S_m$ be the conjuncts of  $\phi_\KB$ 
having the form of $\flqsr$-purely universal formulae. For each $S_i \defAs (\forall z_1^i) \ldots (\forall z_{n_i}^i) \chi_i$, with $i=1,\ldots,m$, we put\\[.4cm]
\centerline{$Exp(S_i) \defAs \underset{ \{x_{a_1}, \ldots, x_{a_{n_i}}\} \subseteq \varz(\phi_{\KB})}{\bigwedge} S_i \{z_1^i / x_{a_1}, \ldots, z^i_{n_i} / x_{a_{n_i}} \}$.}\\[.3cm]
Let us also define the \emph{expansion} $\Phi_\KB$ of $\phi_\KB$ by putting
\begin{equation}\label{phikb}
\Phi_\KB \defAs \{ F_j : i=1,\ldots,k \} \cup \overset{m}{ \underset{i=1}{\bigcup}} Exp(S_i)\,,
\end{equation}
where $F_1, \ldots, F_k$ are the conjuncts of  $\phi_\KB$ having the form of $\flqsr$-quantifier free atomic formulae.

To prepare for Procedure \textit{HOCQA-$\shdlssx$}  to be described next, a brief introduction on the \ke\space system is in order (see \cite{dagostino1999} for a detailed overview of \ke).  \ke\space is a refutation system inspired to Smullyan's semantic tableaux \cite{smullyan1995first}. The main characteristic distinguishing  \ke\space from the latter is the introduction of an analytic cut rule (PB-rule) that permits to reduce inefficiencies of semantic tableaux. In fact, firstly, the classic tableau system can not represent the use of auxiliary lemmas in proofs; secondly, it can not express the bivalence of classical logic. Thirdly, it is extremely inefficient, as witnessed by the fact that it can not polynomially simulate the truth-tables. None of these anomalies occurs if the cut rule is permitted. 
 For these reasons, Procedure  \textit{HOCQA-$\shdlssx$}  constructs a complete \ke\space $\T_{\KB}$ for the expansion $\Phi_\KB$ of $\phi_\KB$ (cf. (\ref{phikb})), representing the saturation of the $\shdlssx$-knowledge base $\KB$. 


Let $\Phi \defAs \{ C_1,\ldots, C_p\}$ be a collection of disjunctions of $\flqsr$-quantifier free atomic formulae of level $0$ of the types: $x =y$, $x \in X^1$, $\langle x,y\rangle \in X^3$. $\mathcal{T}$ is a \textit{\ke}\space for $\Phi$ if there exists a finite sequence $\mathcal{T}_1, \ldots, \mathcal{T}_t$ such that (i) $\mathcal{T}_1$ is a one-branch tree consisting of the sequence $C_1,\ldots, C_p$, (ii) $\mathcal{T}_t = \mathcal{T}$, and (iii) for each $i<t$, $\mathcal{T}_{i+1}$ is obtained from $\mathcal{T}_i$ either by an application of one of the  rules in Fig. \ref{exprule} or by applying a substitution $\sigma$ to a branch $\vartheta$ of $\mathcal{T}_i$ (in particular, the substitution $\sigma$ is applied to each formula $X$ of $\vartheta$; the resulting branch will be denoted with $\vartheta\sigma$). The set of formulae $\seq \defAs \{ \overline{\beta}_1,\ldots,\overline{\beta}_n\} \setminus \{\overline{\beta}_i\}$ occurring as premise in the E-rule contains the  complements of all the components of the formula $\beta$ with the exception of the component $\beta_i$.  

\vspace{-.2cm}
\begin{figure}
{{\small
\begin{center}
\begin{minipage}[h]{5cm}
$\infer[\textbf{E-Rule}]
{\beta_i}{\beta_1 \vee \ldots \vee \beta_n & \quad \seq}$\\[.1cm]
{ where $\seq \defAs \{ \overline{\beta}_1,...,\overline{\beta}_n\} \setminus \{\overline{\beta}_i\}$,}\\[-.1cm] { for $i=1,...,n$}
\end{minipage}~~~~~~~~~~~
\begin{minipage}[h]{2.5cm}
\vspace{0.1cm}
$\infer[\textbf{PB-Rule}]
{A~~|~~\overline{A}}{}$\\[.1cm]
{ with $A$ a literal}
\end{minipage}
\end{center}
\vspace{-.2cm}
}

\caption{\label{exprule}Expansion rules for the \ke.}
}
\end{figure}


Let $\mathcal{T}$ be a \ke. A branch $\vartheta$ of $\mathcal{T}$  is \textit{closed} if it contains either both $A$ and $\neg A$, for some formula $A$, or a literal of type $\neg(x = x)$. Otherwise, the branch is \textit{open}. A \ke\space is \emph{closed} if all its branches are closed. A formula $\beta_1 \vee \ldots \vee \beta_n$ is \textit{fulfilled} in a branch $\vartheta$, if $\beta_i$ is in $\vartheta$, for some $i=1,\ldots,n$. A branch $\vartheta$ is \textit{fulfilled} if every formula $\beta_1 \vee \ldots \vee \beta_n$ occurring in $\vartheta$ is fulfilled. 
A branch $\vartheta$ is \textit{complete} if either it is closed or it is open, fulfilled, and it does not contain any literal of type $x=y$, where $x$ and $y$ are distinct variables. A \ke\space is \textit{complete} (resp., \emph{fulfilled}) if all its  branches are complete (resp., fulfilled or closed).  

 A $\flqsr$-interpretation $\M$ \emph{satisfies} a branch $\vartheta$ of a \ke\space (or, equivalently, $\vartheta$ \emph{is satisfied} by $\M$), and we write $\M \models \vartheta$, if $\M \models X$ for every formula $X$ occurring in $\vartheta$. 
 A $\flqsr$-interpretation $\M$ satisfies a  \ke\space $\mathcal{T}$ (or, equivalently, $\mathcal{T}$ \emph{is satisfied} by $\M$), and we write $\M \models \mathcal{T}$, if $\M$ satisfies a branch $\vartheta$ of $\mathcal{T}$. 
A branch $\vartheta$ of a \ke\space $\mathcal{T}$ is \emph{satisfiable} if there exists a $\flqsr$-interpretation $\M$ that satisfies $\vartheta$. A \ke\space is satisfiable if at least one of its branches is satisfiable.

Let $\vartheta$ be a branch of a \ke. We denote with $<_{\vartheta}$ an arbitrary but fixed total order on the variables in $\mathsf{Var}_{0}(\vartheta)$.

Procedure \textit{HOCQA}-$\shdlssx$ takes care of literals of type $x=y$ occurring in the branches of $\T_\KB$ by constructing, for each open and fulfilled branch $\vartheta$ of $\T_\KB$ a substitution $\sigma_{\vartheta}$ such that $\vartheta\sigma_{\vartheta}$ does not contain literals of type $x=y$ with distinct $x,y$. Then, for every open and complete branch $\vartheta':=\vartheta\sigma_{\vartheta}$ of $\T_{\KB}$, Procedure \textit{HOCQA}-$\shdlssx$ constructs a decision tree $\DT_{\vartheta'}$ such that every maximal branch of $\DT_{\vartheta'}$ induces a substitution $\sigma'$ 
such that $\sigma_{\vartheta}\sigma'$ belongs to the answer set of $\psi_{Q}$ with respect to $\phi_{\KB}$. $\DT_{\vartheta'}$ is defined as follows.

Let $d$ be the number of literals in $\psi_Q$. Then $\DT_{\vartheta'}$ is a finite labelled tree of depth $d+1$ whose labelling satisfies the following conditions, for $i=0,\ldots,d$:
\begin{itemize}[itemsep=0.2cm]
\item[(i)]  every node of $\DT_{\vartheta'}$ at level $i$ is labelled with $(\sigma'_i, \psi_Q\sigma_{\vartheta}\sigma'_i)$; in particular, the root is labelled with
$(\sigma'_0, \psi_Q\sigma_{\vartheta}\sigma'_0)$, where $\sigma'_0$ is the empty substitution; 

\item[(ii)] {if a node at level $i$ is labelled with $(\sigma'_i, \psi_Q\sigma_{\vartheta}\sigma'_i)$, then its $s$ successors, with $s >0$, are labelled with $\big(\sigma'_i\varrho^{q_{i+1}}_1, \psi_Q\sigma_{\vartheta}(\sigma'_i\varrho^{q_{i+1}}_1)\big),\ldots,\big(\sigma'_i\varrho^{q_{i+1}}_s, \psi_Q\sigma_\vartheta(\sigma'_i\varrho^{q_{i+1}}_s)\big)$, 
where $q_{i+1}$ is the $(i+1)$-st conjunct of $\psi_Q\sigma_{\vartheta}\sigma'_i$ and $\mathcal{S}_{q_{i+1}}=\{\varrho^{q_{i+1}}_1, \ldots, \varrho^{q_{i+1}}_s  \}$ is the collection of the substitutions $\varrho = \{v_1/o_1, \ldots, v_n/ o_n,  c_1/C_1, \ldots, c_m/ C_m,\\ r_1/R_1, \ldots, r_k/ R_k, p_1/P_1, \ldots, p_h/ P_h \}$, with $\{v_1, \ldots, v_n\} = \varz(q_{i+1})$,}{\unskip\parfillskip 0pt \par} $\{c_1,\ldots, c_m\} = \varu(q_{i+1})$, and $\{p_1,\ldots, p_h,r_1,\ldots, r_k\} = \vart(q_{i+1})$, such that $t=q_{i+1}\varrho$, for some  literal $t$ on $\vartheta'$. If $s = 0$, the node labelled with $(\sigma'_i, \psi_Q\sigma_{\vartheta}\sigma'_i)$ is a leaf node and, if $i = d$,  $\sigma_\vartheta\sigma'_i$ is added to $\Sigma'$. 
\end{itemize}

We are ready to define Procedure \textit{HOCQA}-$\shdlssx$. 

\smallskip

{\small
\begin{algorithmic}[1]
\Procedure{\textit{HOCQA}-$\shdlssx$}{$\psi_Q$,$\phi_\KB$};

\State $\Sigma'$ := $\emptyset$;

\State - let $\Phi_\KB$ be the expansion of $\phi_\KB$ (cf.\ (\ref{phikb}));

\State $\T_{\KB}$ := $\Phi_\KB$;
\While{$\T_{\KB}$ is not fulfilled}
\parState{- select a not fulfilled open branch $\vartheta$ of $\T_{\KB}$ and a not fulfilled formula\\ \hspace*{.2cm}$\beta_1 \vee \ldots \vee \beta_n$ in $\vartheta$;}
\If{$\seqnj$ is in $\vartheta$, for some $j \in \{1,\ldots,n\}$}
    \State - apply the E-Rule to $\beta_1 \vee \ldots \vee \beta_n$ and $\seqnj$ on $\vartheta$;
\Else 
    \parState{- let $B^{\overline{\beta}}$ be the collection of the formulae $\overline{\beta}_1,\ldots,\overline{\beta}_n$ present in $\vartheta$
and let\\ \hspace*{.15cm} $h$ be the lowest index such that $\overline{\beta}_h \notin B^{\overline{\beta}}$;}
    \parState{- apply the PB-rule to $\overline{\beta}_h$ on $\vartheta$;}
\EndIf;
\EndWhile;
\parWhile{$\T_\KB$ has open branches containing literals of type $x = y$, with distinct $x$ and $y$}
\State{- select such an open branch $\vartheta$ of $\T_\KB$;}

\State $\sigma_{\vartheta} := \epsilon$ (where $\epsilon$ is the empty substitution);

\State $\mathsf{Eq}_{\vartheta} := \{ \mbox{literals of type $x = y$, occurring in $\vartheta$}\}$;

\While{$\mathsf{Eq}_{\vartheta}$ contains $x = y$, with distinct $x$, $y$}

    \State - select a literal $x = y$ in $\mathsf{Eq}_{\vartheta}$, with distinct $x$, $y$;
    
    \State $z := \min_{<_{\vartheta}}(x,y)$;
    
    \State $\sigma_{\vartheta} := \sigma_{\vartheta} \cdot \{x/z, y/z\}$;
    
    \State $\mathsf{Eq}_{\vartheta} := \mathsf{Eq}_{\vartheta}\sigma_{\vartheta}$;
\EndWhile;

\State $\vartheta := \vartheta\sigma_{\vartheta}$;

\If{$\vartheta$ is open}
	\State - initialize $\mathcal{S}$ to the empty stack;

	\State - push $(\epsilon, \psi_Q\sigma_\vartheta)$ in $\mathcal{S}$;
	
		\While{$\mathcal{S}$ is not empty}
		\State - pop $(\sigma', \psi_Q\sigma_\vartheta\sigma')$ from $\mathcal{S}$;
		
		\If{$\psi_Q\sigma_\vartheta\sigma' \neq \lambda$}
			\State - let $q$ be the leftmost conjunct of $\psi_Q\sigma_\vartheta\sigma'$;
			
			\State $\psi_Q\sigma_\vartheta\sigma':= \psi_Q\sigma_\vartheta\sigma'$ deprived of $q$;
			\State $Lit^{M}_Q := \{ t \in \vartheta : t=q\rho$, for some substitution $\rho \}$;
			
			\While{$Lit^{M}_Q$ is not empty}
				\State - let $t \in Lit^{M}_Q$, $t=q\rho$;
				
				\State $Lit^{M}_Q := Lit^{M}_Q \setminus \{t\}$;
				
				\State - push $(\sigma'\rho, \psi_Q\sigma_\vartheta\sigma'\rho)$ in $\mathcal{S}$;
			\EndWhile;
		\Else
			\State $\Sigma'$ := $\Sigma' \cup \{\sigma_\vartheta\sigma'\}$;
		\EndIf;
	\EndWhile;

\EndIf;
\EndparWhile;

%
%
%
%
%
%
%
%

\State \Return $(\T_\KB, \Sigma')$;
\EndProcedure;
\end{algorithmic}
}

\normalsize

For each open branch $\vartheta$ of $\T_\KB$, Procedure \textit{HOCQA}-$\shdlssx$ computes the corresponding $\DT_{\vartheta}$ by constructing a stack of its nodes. Initially the stack contains the root node $(\epsilon,\psi_Q\sigma_\vartheta)$ of  $\DT_{\vartheta}$, as defined in condition (i). Then, iteratively, the following steps are executed. An element $(\sigma', \psi_Q\sigma_\vartheta\sigma')$ is popped out of the stack. If the last literal of the query $\psi_Q$ has not been reached, the successors of the current node are computed according to condition (ii) and inserted in the stack. Otherwise the current node must have the form $(\sigma',\lambda)$ and the substitution $\sigma_\vartheta\sigma'$  is inserted in $\Sigma'$.

%
Correctness of Procedure \textit{HOCQA}-$\shdlssx$ follows from Theorems \ref{teo:correctness} and \ref{teo:completeness}, which show that $\phi_\KB$ is satisfiable if and only if $\T_{\KB}$ is a non-closed \ke, and from Theorem \ref{teo:proc2corr}, which shows that the set $\Sigma'$ coincides with the HO answer set of $\psi_Q$ w.r.t. $\phi_\KB$. Theorems \ref{teo:correctness}, \ref{teo:completeness}, and \ref{teo:proc2corr} are stated below. 
In particular, Theorem \ref{teo:correctness}, requires the following technical lemmas. 
\begin{lemma}\label{lemma:invariant}
	Let $\vartheta$ be a branch of $\T_{\KB}$ selected at step $15$ of Procedure \textit{HOCQA}-$\shdlssx$($\psi_{Q}$,$\phi_{\KB}$), let $\sigma_{\vartheta}$ be the associated substitution constructed during the execution of the while-loop 18--23, and let $\M = (D,M)$ be a \flqsr-interpretation satisfying $\vartheta$. Then 
	\begin{equation}\label{eq:invariant}
	Mx = Mx\sigma_{\vartheta}, \mbox{ for every } x \in \mathsf{Var}_0(\vartheta),
	\end{equation}
	is an invariant of the while-loop 18--23. 
\end{lemma}

\begin{proof}
We prove the thesis by induction on the number $i$ of iterations of the while loop 18--23 of the procedure \textit{HOCQA}-$\shdlssx$($\psi_{Q}$,$\phi_{\KB}$). For simplicity we indicate with $\sigma_{\vartheta}^{(i)}$ and with $Eq_{\sigma_{\vartheta}}^{(i)}$ the substitution $\sigma_{\vartheta}$ and the set $Eq_{\sigma_{\vartheta}}$calculated at iteration $i \geq 0$, respectively. 

If $i = 0$, $\sigma_{\vartheta}^{(0)}$ is the empty substitution $\epsilon$ and thus (\ref{eq:invariant}) trivially holds. 

Assume by inductive hypothesis that (\ref{eq:invariant}) holds at iteration $i \geq 0$. We want to prove that (\ref{eq:invariant}) holds at iteration $i+1$. 

At iteration $i +1$, $\sigma_{\vartheta}^{(i+1)} = \sigma_{\vartheta}^{(i)} \cdot \{x/z,y/z\}$, where $z = \min_{<_{\vartheta}}\{x,y\}$ and $x = y$ is a literal in $Eq_{\sigma_{\vartheta}}^{(i)}$, with distinct $x,y$. We assume, without loss of generality, that $z$ is the variable $x$ (an analogous proof can be carried out assuming that $z$ is the variable $y$). By inductive hypothesis $Mw = Mw\sigma_{\vartheta}^{(i)}$, for every $w \in \mathsf{Var}_0(\vartheta)$. If $w\sigma_{\vartheta}^{(i)} \in \mathsf{Var}_0(\vartheta)\setminus \{y\}$, plainly $w\sigma_{\vartheta}^{(i)}$ and $w\sigma_{\vartheta}^{(i+1)}$ coincide and thus 
$Mw\sigma_{\vartheta}^{(i)} = Mw\sigma_{\vartheta}^{(i+1)}$. Since $Mw = Mw\sigma_{\vartheta}^{(i)}$, it follows that $Mw = Mw\sigma_{\vartheta}^{(i+1)}$. 

If $w\sigma_{\vartheta}^{(i)}$ coincides with $y$ we reason as follows. At iteration $i+1$ variables $x,y$ are considered because the literal $x=y$ is selected from $Eq_{\sigma_{\vartheta}}^{(i)}$. If $x=y$ is a literal belonging to $\vartheta$, then $Mx = My$. Since $w\sigma_{\vartheta}^{(i)}$ coincides with $y$, $w\sigma_{\vartheta}^{(i+1)}$ coincides with $x$, $My = Mx$, and by inductive hypothesis 
$Mw = Mw\sigma_{\vartheta}^{(i)}$, it holds that $Mw = Mw\sigma_{\vartheta}^{(i+1)}$. If $x  = y$ is not a literal occurring in $\vartheta$, then $\vartheta$ must contain a literal $x' = y'$ such that, at iteration $i$, $x$ coincides with $x'\sigma_{\vartheta}^{(i)}$ and $y$ coincides with $y'\sigma_{\vartheta}^{(i)}$. Since $Mx' = My'$ and, by inductive hypothesis, $Mx' = Mx'\sigma_{\vartheta}^{(i)}$, and $My' = My'\sigma_{\vartheta}^{(i)}$, it holds that $Mx = My$, and thus, reasoning as above,  $Mw = Mw\sigma_{\vartheta}^{(i+1)}$. Since (\ref{eq:invariant}) holds at each iteration of the while loop, it is an invariant of the loop as we wished to prove.\qed
\end{proof}
\begin{lemma}\label{lemma:correctness}
	Let $\T_0,\ldots,\T_h$ be a sequence of \ke x such that $\T_0 = \Phi_{\KB}$, and $\T_{i+1}$ is obtained from $\T_i$ by applying either the rule of step 8, or the rule of step 10, or the  substitution of step 24 of Procedure \textit{HOCQA}-$\shdlssx$($\psi_{Q}$,$\phi_{\KB}$), for $i = 1,\ldots,h-1$. 
	If $\T_i$ is satisfied by a \flqsr-interpretation $\M$, then $\T_{i+1}$ is satisfied by $\M$ as well, for $i = 1,\ldots,h-1$.
\end{lemma}  
\begin{proof}
Let $\M=(D,M)$ be a \flqsr-interpretation satisfying $\T_i$. Then $\M$ satisfies a branch $\bar{\vartheta}$ of $\T_i$. In case the branch $\bar{\vartheta}$ is different from the branch selected at step 6, if the E-rule (step 8) or the PB-rule (10) is applied, or at step 3, if a substitution for handling equalities (step 14) is applied, $\bar{\vartheta}$ belongs to $\T_{i+1}$ and therefore $\T_{i+1}$ is satisfied by $\M$. 
In case $\bar{\vartheta}$ is the branch selected and modified to obtain $\T_{i+1}$, we have to consider the following distinct cases. 
\begin{itemize}
\item $\bar{\vartheta}$ has been selected at step 6 and thus it is an open branch not yet fulfilled. Then, if step 8 is executed, the E-rule is applied to a not fulfilled formula $\beta_1 \vee \ldots \vee \beta_n$ and to the set of formulae $\seqnj$ on the branch $\bar{\vartheta}$ generating the new branch $\bar{\vartheta'} := \bar{\vartheta} ; \beta_i$. Plainly, if $\M \models \bar{\vartheta}$, $\M \models \beta_1 \vee \ldots \vee \beta_n$, $\M \models \seqnj$ and, as a consequence, $\M \models \beta_i$. Thus $\M \models \bar{\vartheta'}$ and finally, $\M$ satisfies $\T_{i+1}$. If step 10 is performed, the PB-rule is applied on $\bar{\vartheta}$ originating the  branches  (belonging to $\T_{i+1}$) $\bar{\vartheta'} := \bar{\vartheta} ; \overline{\beta}_h$ and $\bar{\vartheta''} := \bar{\vartheta} ; \beta_h$. Since either $\M \models \beta_h$ or  $\M \models \overline{\beta}_h$, it holds that either $\M \models \bar{\vartheta'}$ or $\M \models \bar{\vartheta''}$. Thus $\M$ satisfies $\T_{i+1}$, as we wished to prove. 
\item $\bar{\vartheta}$ has been selected at step 14 and thus it is an open and fulfilled branch not yet complete. Once step $24$ is executed the new branch $\bar{\vartheta} \sigma_{\bar{\vartheta}}$ is generated. Since $\M \models \bar{\vartheta}$ and, by Lemma \ref{lemma:invariant}, $Mx = Mx\sigma_{\bar{\vartheta}}$, for every $x \in \mathsf{Var}_0(\bar{\vartheta})$, it holds that $\M \models \bar{\vartheta} \sigma_{\bar{\vartheta}}$ and that $\M$ satisfies $\T_{i+1}$.  Thus the thesis follows.\qed 
\end{itemize} 
\end{proof}
%

%

Then we have:

\begin{theorem}\label{teo:correctness}
	If $\phi_{\KB}$ is satisfiable, then $\T_{\KB}$ is not closed.
\end{theorem}
\begin{proof}
Let us assume by contradiction that $\T_{\KB}$ is closed. Since $\Phi_{\KB}$ is satisfiable, there exists a $\flqsr$-interpretation $\M$ satisfying every formula of $\Phi_{\KB}$. Thanks to Lemma \ref{lemma:correctness}, any \ke\space for $\Phi_{\KB}$ obtained by applying either step 8, or step 10, or step 24 of the procedure \textit{HOCQA}-$\shdlssx$, is satisfied by $\M$. Thus $\T_{\KB}$ is satisfied by $\M$ as well. In particular, there exists a branch $\vartheta_c$ of $\T_{\KB}$ satisfied by $\M$. Since $\T_{\KB}$ is closed, by the absurd hypothesis,  the branch $\vartheta_c$ is closed as well and thus, by definition, it contains either both $A$ and $\neg A$, for some formula $A$, or a literal of type $\neg (x = x)$. $\vartheta$ is satisfied by $\M$ and thus, either $\M \models A$ and $\M \models \neg A$ or $\M \models \neg (x = x)$. Absurd. Thus, we have to admit that the \ke\space $\T_{\KB}$ is not closed.\qed  
\end{proof}

\begin{theorem}\label{teo:completeness}
	If $\T_{\KB}$ is not closed, then  $\phi_{\KB}$ is satisfiable.
\end{theorem}
\begin{proof}
\emph{Proof}. Since $\T_{\KB}$ is not closed, there exists a branch $\vartheta'$ of $\T_{\KB}$ which is open and complete. 
The branch $\vartheta'$ is obtained during the execution of the procedure \textit{HOCQA}-$\shdlssx$ from an open fulfilled branch $\vartheta$ by applying to $\vartheta$ the substitution $\sigma_{\vartheta}$ constructed during the execution of step 14 of the procedure. Thus, $\vartheta' = \vartheta\sigma_{\vartheta}$. Since each formula of $\Phi_{\KB}$ occurs in $\vartheta$, showing that $\vartheta$ is satisfiable is enough to prove that $\Phi_{\KB}$ is satisfiable. 

Let us construct a \flqsr-interpretation $\M_{\vartheta}=(D_{\vartheta},M_{\vartheta})$ satisfying every formula $X$ occurring in $\vartheta$ and thus $\Phi_{\KB}$. $\M_{\vartheta}=(D_{\vartheta},M_{\vartheta})$ is defined as follows. 
\begin{itemize}
\item $D_{\vartheta} \defAs \{x\sigma_{\vartheta} : x \in \mathsf{Var}_0(\vartheta) \}$; 
\item $M_{\vartheta} x \defAs x\sigma_{\vartheta}$, $x \in \mathsf{Var}_0(\vartheta)$; 
\item $M_{\vartheta} X^1 \defAs \{x\sigma_{\vartheta} : x \in X^1 \mbox{ occurs in } \vartheta\}$, $X^1 \in \mathsf{Var}_1(\vartheta)$;  
\item $M_{\vartheta} X^3 \defAs \{\langle x\sigma_{\vartheta}, y\sigma_{\vartheta} \rangle : \langle x, y\rangle \in X^3  \mbox{ occurs in } \vartheta \}$, $X^3 \in \mathsf{Var}_3(\vartheta)$. 
\end{itemize}


In what follows we show that $\M_{\vartheta}$ satisfies each formula in $\vartheta$. Our proof is carried out by induction on the structure of formulae and cases distinction. 
Let us consider, at first, a literal $x = y$ occurring in $\vartheta$. By the construction of $\sigma_{\vartheta}$ described in the procedure, $x\sigma_{\vartheta}$ and $y\sigma_{\vartheta}$ have to coincide. Thus $M_{\vartheta}x = x\sigma_{\vartheta} = y\sigma_{\vartheta} = M_{\vartheta} y$ and then $\M_{\vartheta} \models x=y$.

Next we consider a literal $\neg (z = w)$ occurring in $\vartheta$. If $z\sigma_{\vartheta}$ and $w\sigma_{\vartheta}$ coincide, namely they are the same variable, then the branch $\vartheta' = \vartheta\sigma_{\vartheta}$ must be a closed branch against our hypothesis. Thus $z\sigma_{\vartheta}$ and $w\sigma_{\vartheta}$ are distinct variables and therefore $M_{\vartheta} z= z\sigma_{\vartheta} \neq w\sigma_{\vartheta} = M_{\vartheta}w$, then $\M_{\vartheta} \not\models z = w$ and finally $\M_{\vartheta} \models \neg(z = w)$, as we wished to prove. 
  
Let $x \in X^1$ be a literal occurring in $\vartheta$. By the definition of $M_{\vartheta}$, $x\sigma_{\vartheta} \in M_{\vartheta}X^1$, namely $M_{\vartheta} x \in M_{\vartheta}X^1$ and thus $\M_{\vartheta} \models x \in X^1$ as desired. 
If $\neg(y \in X^1)$ occurs in $\vartheta$, then $y\sigma_{\vartheta}\notin M_{\vartheta}X^1$. Assume, by contradiction that $y\sigma_{\vartheta}\in M_{\vartheta}X^1$. Then there is a literal $z \in X^1$ in $\vartheta$ such that $z\sigma_{\vartheta}$ and $y\sigma_{\vartheta}$ coincide. In this case the branch $\vartheta'$, obtained from $\vartheta$ applying the substitution  $\sigma_{\vartheta}$ would be closed, contradicting the hypothesis. Thus $y\sigma_{\vartheta}\notin M_{\vartheta}X^1$ implies that $M_{\vartheta}y \notin M_{\vartheta}X^1$, that  $\M_{\vartheta} \not\models y \in X^1$, and finally that $\M_{\vartheta} \models \neg(y \in X^1)$. 

If $\langle x,y\rangle \in X^3$ is a literal on $\vartheta$, then by definition of $M_{\vartheta}$, $\langle x\sigma_{\vartheta}, y\sigma_{\vartheta}\rangle \in M_{\vartheta}X^3$, that is $\langle M_{\vartheta}x, M_{\vartheta}y\rangle \in M_{\vartheta}X^3$, and thus 
$\M_{\vartheta} \models \langle x,y\rangle \in X^3$. 

Let $\neg(\langle z,w\rangle \in X^3)$ be a literal occurring on $\vartheta$. Assume that $\langle z\sigma_{\vartheta},w\sigma_{\vartheta}\rangle \in M_{\vartheta}X^3$. Then a literal $\langle z',w'\rangle \in X^3$ occurs in $\vartheta$ such that $z\sigma_{\vartheta}$ coincides with $z'\sigma_{\vartheta}$ and that  $w\sigma_{\vartheta}$ coincides with $w'\sigma_{\vartheta}$. But then the branch $\vartheta' = \vartheta\sigma_{\vartheta}$ would be closed contradicting the hypothesis. Thus we have to admit that $\langle z\sigma_{\vartheta},w\sigma_{\vartheta}\rangle \notin M_{\vartheta}X^3$, that is $\langle M_{\vartheta}z,M_{\vartheta}w\rangle \notin M_{\vartheta}X^3$. Thus $\M_{\vartheta} \not\models \langle x,y\rangle \in X^3$ and finally $\M_{\vartheta} \models \neg(\langle x,y\rangle \in X^3)$. 

Let $\beta = \beta_1 \vee \ldots \vee \beta_k$ be a disjunction of literals in $\vartheta$. Since $\vartheta$ is fulfilled, $\beta$ is fulfilled too and, therefore, $\vartheta$ contains a disjunct $\beta_i$, for some $i \in \{1,\ldots,k\}$ of $\beta$. By inductive hypothesis $\M_{\vartheta} \models \beta_i$ and thus $\M_{\vartheta} \models \beta$. 

We have shown that $\M_{\vartheta}$ satisfies each formula in $\vartheta$ and, in particular the formulae in $\Phi_{\KB}$. It turns out that $\Phi_{\KB}$ is satisfiable as we wished to prove. \qed
\end{proof}


It is easy to check that the $\flqsr$-interpretation $\M_{\vartheta}$ defined in Theorem \ref{teo:completeness} satisfies $\phi_{\KB}$, a collection of $\flqsr$-purely universal formulae and of $\flqsr$-quantifier free atomic formulae corresponding to a $\shdlssx$-knowledge base $\KB$ and, therefore, that the following corollary holds. 

\begin{corollary}
	If $\T_{\KB}$ is not closed, then $\phi_{\KB}$ is satisfiable.
\end{corollary}

In what follows, we state also a technical lemma which is needed in the proof of Theorem \ref{teo:proc2corr}. 

\begin{lemma}\label{lemma:proc2}
	Let $\psi_Q \defAs q_1 \wedge \ldots \wedge q_d$ be a HO $\flqsr$-conjunctive query, let $(\T_\KB, \Sigma')$ be the output of \textit{HOCQA}-$\shdlssx$($\psi_Q$,$\phi_\KB$), and let $\vartheta$ be an open and complete branch of $\T_\KB$. Then, for any substitution $\sigma$, we have 
	$$\sigma \in \Sigma' \iff \{ q_1 \sigma, \ldots, q_d\sigma \} \subseteq \vartheta\,.$$
\end{lemma}
\begin{proof} If $\sigma' \in \Sigma'$, then $\sigma'=\sigma_\vartheta\sigma'_1$ and the decision tree $\DT_{\vartheta'}$ contains a branch $\eta$ of length $d+1$ having as leaf $(\sigma'_1, \lambda)$. Specifically, the branch $\eta$ is constituted by the nodes 

 $(\epsilon, q_1\sigma_\vartheta \wedge \ldots \wedge q_d\sigma_\vartheta )$, $( \rho^{(1)}, q_2\sigma_\vartheta\rho^{(1)} \wedge \ldots \wedge q_d\sigma_\vartheta\rho^{(1)} )$, $\ldots$, $( \rho^{(1)} \ldots \rho^{(d)}, \lambda)$,

 \noindent and hence $\sigma'= \sigma_\vartheta \rho^{(1)} \ldots \rho^{(d)}$. 
 
\noindent Consider the node

 $ (\rho^{(1)} \ldots \rho^{(i+1)}, q_{i+2}\sigma_\vartheta \rho^{(1)} \ldots \rho^{(i+1)} \wedge \ldots \wedge q_d\sigma_\vartheta \rho^{(1)} \ldots \rho^{(i+1)})$ 
 
 \noindent constructed from the father node 
 
 $(\rho^{(1)} \ldots \rho^{(i+1)}, q_{i+1}\sigma_\vartheta \rho^{(1)} \ldots \rho^{(i)} \wedge \ldots \wedge q_d\sigma_\vartheta \rho^{(1)} \ldots \rho^{(i)})$ 
 
 \noindent putting $q_{i+1}\sigma_\vartheta \rho^{(1)} \ldots \rho^{(i+1)} = t$, for some $t \in \vartheta'$. Since $q_{i+1}\sigma_\vartheta \rho^{(1)} \ldots \rho^{(i+1)}$ is a ground literal,  $q_{i+1}\sigma_\vartheta \rho^{(1)} \ldots \rho^{(i+1)}$ coincides with $q_{i+1}\sigma'$, then $q_{i+1} \sigma'=t$, and hence $q_{i+1}\sigma' \in \vartheta'$. Given the generality of $i=0, \ldots, d-1$, $\{ q_1\sigma', \ldots, q_d \sigma' \} \subseteq \vartheta'$ as we wished to prove.

 We now prove the second part of the lemma. We show that the decision tree  $\DT_{\vartheta'}$ constructed by Procedure \textit{HOCQA}-$\shdlssx$($\psi_{Q}$,$\phi_{\KB}$) has a branch $\eta$ of length $d+1$ having as leaf the node $(\sigma'_1,\lambda)$, with $\sigma_\vartheta\sigma'_1=\sigma' \in \Sigma'$. Since by hypothesis  $\vartheta'= \vartheta\sigma_\vartheta$, the root of the decision tree $\DT_{\vartheta'}$ is the node $(\epsilon, q_1\sigma_\vartheta \wedge \ldots \wedge q_d\sigma_\vartheta )$. At step $i$, the procedure selects a literal $q^{(i)}$, namely $q_{i}\sigma_{\vartheta}\rho^{(1)}\ldots\rho^{(i-1)}$, and finds a substitution $\rho^{(i)}$ such that $q_{i}\sigma_{\vartheta}\rho^{(1)}\ldots\rho^{(i)}$ coincides with $q_{i}\sigma'$. Then, the procedure constructs the node 
 
 $(\rho^{(1)} \ldots \rho^{(i)}, q_{i+1}\sigma_\vartheta \rho^{(1)} \ldots \rho^{(i)} \wedge \ldots \wedge q_d\sigma_\vartheta \rho^{(1)} \ldots \rho^{(i)})$  
 
 At step $d-1$, the procedure constructs the leaf node $(\rho^{(1)}\ldots \rho^{(d)},\lambda)$, that is $(\sigma'_1,\lambda)$, as we wished to prove.\qed


\end{proof}
%
\begin{theorem}\label{teo:proc2corr}
	Let $\Sigma'$ be the set of substitutions returned by Procedure \textit{HOCQA}-$\shdlssx$($\psi_Q$, $\phi_\KB$). Then $\Sigma'$ is the HO answer set of $\psi_Q$ w.r.t. $\phi_\KB$.
\end{theorem}
\begin{proof}
To prove the theorem we show that the following two assertions hold.
\begin{enumerate}
\item If $\sigma' \in \Sigma'$, then $\sigma'$ is an element of the HO answer set of $\psi_Q$ w.r.t. $\phi_\KB$.
\item If $\sigma'$ is a substitution of the HO answer set of $\psi_Q$ w.r.t. $\phi_\KB$, then $\sigma' \in \Sigma'$.
\end{enumerate}
We prove assertion (1) as follows. Let $\sigma' \in \Sigma'$ and $\vartheta'=\vartheta\sigma_\vartheta$ an open and complete branch of $\T_\KB$ such that $\DT_{\vartheta'}$ contains a branch $\eta$ of $d+1$ nodes whose leaf is labelled $\langle \sigma_1', \lambda \rangle$, where $\sigma_1'$ is a substitution such that $\sigma' = \sigma_{\vartheta}\sigma_1'$. By Lemma \ref{lemma:proc2}, $\{  q_1\sigma', \ldots, q_d\sigma' \} \subseteq \vartheta'$. Let $\M_\vartheta$ be a $\flqsr$-interpretation constructed as shown in Theorem \ref{teo:completeness}. We have that $\M_\vartheta \models q_{i}\sigma'$, for $i = 1,\ldots,d$ because $\{q_1\sigma', \ldots, q_d\sigma' \} \subseteq \vartheta'$ holds. Thus $\M_\vartheta \models \psi_Q\sigma'$, and since $\M_\vartheta \models \phi_\KB$,  $\M_\vartheta \models \phi_\KB \wedge \psi_Q\sigma'$ holds.  Hence $\sigma'$ is a substitution of the answer set of $\psi_Q$ w.r.t. $\phi_\KB$. To show that assertion (2) holds, let us consider a substitution $\sigma'$ belonging to the answer set of $\psi_Q$ w.r.t. $\phi_\KB$. Then there exists a $\flqsr$-interpretation $\M \models \phi_\KB \wedge \psi_Q\sigma'$.
 Assume by contradiction that $\sigma' \notin \Sigma'$. Then, by Lemma \ref{lemma:proc2}  $\{q_1\sigma, \ldots, q_d\sigma' \} \not\subseteq\vartheta'$, for every open and complete branch $\vartheta'$ of $\T_\KB$. In particular, given any open complete branch $\vartheta'$ of $\T_\KB$, there is an $i \in \{1,\ldots,d\}$ such that $q_i\sigma' \notin \vartheta' = \vartheta\sigma_{\vartheta}$ and thus $\M_\vartheta \not\models q_i\sigma'$.

 By the generality of $\vartheta'=\vartheta\sigma_\vartheta$, it holds that every $\M_\vartheta$ satisfying $\T_\KB$, and thus $\phi_\KB$, does not satisfy $\psi_Q\sigma'$. Since  we can prove that $\M\models \phi_\KB \wedge \psi_Q\sigma'$, for some $\flqsr$-interpretation $\M$, by restricting our interest to the interpretations $\M_\vartheta$ of $\phi_\KB$ defined in the proof of Theorem \ref{teo:completeness}, it turns out that $\sigma'$ is not a substitution belonging to the answer set of $\psi_Q$ w.r.t. $\phi_\KB$, and this leads to a contradiction. Thus we have to admit that assertion (2) holds. Finally, since assertions (1) and (2) hold, $\Sigma'$ and the answer set of $\psi_Q$ w.r.t. $\phi_\KB$ coincide and the thesis holds.\qed

\end{proof}


Termination of Procedure \textit{HOCQA}-$\shdlssx$ is based on the fact that the while-loops 5--13 and 14--44 terminate. 

Concerning termination of the while-loop 5--13, our proof is based on the following two facts. The E-Rule and PB-Rule are applied only to non-fulfilled formulae on open branches and tend to reduce the number of non-fulfilled formulae occurring on the considered branch. In particular, when the E-Rule is applied on a branch $\vartheta$, the number of non-fulfilled formulae on $\vartheta$ decreases. In case of application of the PB-Rule on a formula $\beta = \beta_1 \vee \ldots \vee \beta_n$ on a branch, the rule generates two branches. In one of them the number of non-fulfilled formulae decreases (because $\beta$ becomes fulfilled).  In  the other one the number of non-fulfilled formulae stays constant but the subset $B^{\overline{\beta}}$ of $\{\overline{\beta}_1,\ldots,\overline{\beta}_n\}$ occurring on the branch gains a new element. Once $|B^{\overline{\beta}}|$ gets equal to $n-1$, namely after at most $n-1$ applications of the PB-rule, the E-rule is applied and the formula $\beta = \beta_1 \vee \ldots \vee \beta_n$ becomes fulfilled, thus decrementing the number of non-fulfilled formulae on the branch. Since the number of non-fulfilled formulae on each open branch gets equal to zero after a finite number of steps and the E-rule and PB-rule can be applied only to non-fulfilled formulae on open branches,  the while-loop 5--13 terminates.
Concerning the while-loop 14--44, its termination can be proved by observing that the number of branches of the \ke\space resulting from the execution of the previous while-loop 5--13 is finite and then showing that the internal while-loops 18--23 and 28--42 always terminate. 
Indeed, initially the set $\mathsf{Eq}_{\vartheta}$ contains a finite number of literals of type $x=y$, and $\sigma_\vartheta$ is the empty substitution. It is then enough to show that the number of literals of type $x=y$ in $\mathsf{Eq}_{\vartheta}$, with distinct $x$ and $y$, strictly decreases during the execution of the internal while-loop 18--23. But this follows immediately, since at each of its iterations one put $\sigma_{\vartheta} \defAs \sigma_{\vartheta} \cdot \{x/z, y/z\}$, with $z \defAs \min_{<_{\vartheta}}(x,y)$, according to a fixed total order $<_{\vartheta}$ over the variables of $\mathsf{Var}_{0}(\vartheta)$ and then  the application of $\sigma_\vartheta$ to $\mathsf{Eq}_{\vartheta}$ replaces a literal of type $x=y$ in $\mathsf{Eq}_{\vartheta}$, with distinct $x$ and $y$, with a literal of type $x=x$.

The while loop 28--42 terminates  when the stack $\mathcal{S}$ of the nodes of the decision tree gets empty. Since the query $\psi_Q$ contains a finite number of conjuncts and the number of literals on each open and complete branch of $\T_{\KB}$ is finite, the number of possible matches (namely the size of the set $Lit^M_Q$) computed at step $(C )$ is finite as well. Thus, in particular, the internal while loop 34--38 terminates at each execution. Once the procedure has processed the last conjunct of the query, the set $Lit^M_Q$ of possible matches is empty and thus no element gets pushed in the stack $\mathcal{S}$ anymore. Since the first instruction of the while-loop at step $(i)$ removes an element from $\mathcal{S}$, the stack gets empty after a finite number of ``pops''. Hence Procedure \textit{HOCQA}-$\shdlssx$ terminates, as we wished to prove.  

%


Next, we provide some complexity results.
Let $r$ be the maximum number of universal quantifiers in each $S_i$ ($i=1,\ldots,m$), and put $k \defAs |\varz(\phi_{\KB})|$. Then, each $S_i$ generates at most $k ^r$ expansions. Since the knowledge base contains $m$ such formulae, the number of disjunctions in the initial branch of the \ke\space is bounded by $m \cdot k^r$. Next, let $\ell$ be the maximum number of literals in each $S_i$. Then, the height of the \ke (which corresponds to the maximum size of the models of $\Phi_{\KB}$ constructed as illustrated above) is $\mathcal{O}(  \ell  m  k^r)$ and the number of leaves of the tableau, namely the number of such models of $\Phi_{\KB}$, is $\mathcal{O}(2^{\ell  m  k^r})$. Notice that the construction of $\mathsf{Eq}_{\vartheta}$ and of $\sigma_{\vartheta}$ in the lines 16--23 of Procedure \textit{HOCQA}-$\shdlssx$ takes $\mathcal{O}(  \ell  m  k^r)$ time, for each branch $\vartheta$.

Let $\eta(\T_{\KB})$ and $\lambda(\T_{\KB})$ be, respectively, the height of $\T_{\KB}$ and the number of leaves of $\T_{\KB}$ computed by Procedure \textit{HOCQA}-$\shdlssx$.  Plainly, $\eta(\T_{\KB}) = \mathcal{O}(\ell  m  k^r)$  and $\lambda(\T_{\KB})= \mathcal{O}(2^{\ell  m  k^r})$, as computed above. It is easy to verify that $s=\mathcal{O}(\ell   k^r)$ is the maximum branching of $\DT_\vartheta$. Since the height of $\DT_\vartheta$ is $h$, where $h$ is the number of literals in $\psi_Q$, and the successors of a node are computed in $\mathcal{O}(\ell   k^r)$ time, the number of leaves in $\DT_\vartheta$ is $\mathcal{O}(s^{h})=\mathcal{O}((\ell  k^r)^{h})$ and they are computed  in $\mathcal{O}( s^{h} \cdot  \ell k^r \cdot h) = \mathcal{O}(h \cdot (\ell  k^r)^{(h+1)})$ time. Finally, since we have $\lambda(\T_{\KB})$ of such decision trees, the answer set of $\psi_{Q}$ w.r.t.\ $\phi_\KB$ is computed in time
$\mathcal{O}(h \cdot (\ell  k^r)^{(h+1)} \cdot\lambda(\T_{\KB})) =\mathcal{O}( h \cdot (\ell   k^r)^{(h+1)} \cdot 2^{\ell  m  k^r})$.

 Since the size of $\phi_\KB$ and of $\psi_{Q}$ are  related to those of $\KB$ and of $Q$, respectively (see the proof of Theorem \ref{CQADL} in \cite{RR2017ext} for details on the reduction), the construction of the HO answer set of $Q$ with respect to $\KB$ can be done in double-exponential time. In case $\KB$ contains neither role chain axioms nor qualified cardinality restrictions, the complexity of our \textit{HOCQA} problem is in EXPTIME, since the maximum number of universal quantifiers in $\phi_{\KB}$, namely $r$, is a constant (in particular $r = 3$). The latter complexity result is a clue of the fact that the \textit{HOCQA} problem is intrinsically more difficult than the consistency problem (proved to be NP-complete in \cite{CanLonNicSanRR2015}). This is motivated by the fact that the consistency problem simply requires to guess a model of the knowledge base while the \textit{HOCQA} problem forces the construction of all the models of the knowledge base and to compute a decision tree for each of them.  

We remark that such result compares favourably to the complexity of the usual CQA problem for a wide collection of description logics such as the Horn fragment of $\mathcal{SHOIQ}$ and of $\mathcal{SROIQ}$, respectively, EXPTIME- and 2EXPTIME-complete in combined complexity
(see \cite{Ortiz:2011:QAH:2283516.2283571} for details).

\section{Conclusions and future work}
In this paper we have considered an extension of the CQA problem for the description logic $\shdlssx$ to more general queries on roles and concepts. The resulting problem, called HOCQA, can be instantiated to the most widespread ABox reasoning services such as instance retrieval, role filler retrieval, and instance checking. We have proved the decidability of the HOCQA problem   by reducing it to the satisfiability problem for the set-theoretic fragment $\flqsr$.

We have introduced an algorithm  to compute the HO answer set of a $\flqsr$-formula $\psi_{Q}$ representing a HO $\shdlssx$-conjunctive query $Q$ w.r.t. a $\flqsr$-formula $\phi_{\KB}$ representing a $\shdlssx$ knowledge base. The procedure, called \textit{HOCQA}-$\shdlssx$, is based on the \ke\space system and on  decision trees. It takes as input $\psi_{Q}$ and $\phi_{\KB}$, and yields a \ke\space $\T_{\KB}$ representing the saturation of $\phi_{\KB}$ and the requested HO answer set $\Sigma'$. Procedure \textit{HOCQA}-$\shdlssx$ is proved correct and complete, and some complexity results are provided. Such procedure extends the one introduced in  \cite{ictcs16}
as it allows one to handle HO $\shdlssx$-conjunctive queries. 


We are currently working at the implementation of Procedure \textit{HOCQA}-$\shdlssx$. We plan to increase the efficiency of the expansion rules and 
to extend reasoning with data types. Lastly, we intend to provide a parallel model of the  procedure that we are implementing.

We also plan to increase the expressive power of the set theoretic fragments we are working with. In particular, we intend to define a decidable $n$-level stratified syllogistic allowing to represent an extension of $\shdlssx$ admitting data type groups. 

We also intend to extend the set-theoretic fragment presented in \cite{CanNic2013} with the construct of generalized union and with a restricted form of binary relational composition operator. The latter operator, in particular, turns out to be useful for the set-theoretic representation of various logics. The \ke\space based procedure will be adapted to the new set-theoretic fragments by also making use of the techniques introduced in \cite{CanNicOrl11} and in \cite{CanNicOrl10} in the area of relational dual tableaux. 
On the other hand we think that \ke x could be used in the ambit of relational dual tableaux to improve the performances of relational dual tableau-based decision procedures.

\bibliographystyle{plain}
\bibliography{biblioext} 
\clearpage

\end{document}